\documentclass[11pt,a4paper]{article}
\usepackage[T1]{fontenc}
\usepackage[english]{babel}

\usepackage{ulem}
\usepackage{amsmath, nccmath}
\usepackage{amsfonts}
\usepackage{amsthm}
\usepackage{amssymb}
\usepackage{bbm}

\usepackage[shortlabels]{enumitem}
\usepackage{tikz}
\usepackage{tikz-cd}
\usepackage[colorlinks,linkcolor=blue]{hyperref}
\usepackage{graphicx}

\usepackage{float}
\usepackage{geometry}
\usepackage{subfigure}
\usepackage{multicol}
\usepackage{mathrsfs}
\usepackage{xcolor}
\usepackage{tcolorbox}
\usepackage{framed}
\usepackage{mathtools}
\usepackage{extarrows}

\setlist{noitemsep}
\allowdisplaybreaks

\tcbuselibrary{breakable}
\newcommand{\me}{\textrm{e}}
\newcommand{\ddp}{\displaystyle}
\newcommand{\mm}[1]{\mathrm{#1}}
\newcommand{\mb}[1]{\mathbb{#1}}

\newcommand{\dd}{\mathrm{d}}

\newcommand{\ol}[1]{\overline{#1}}

\newtcolorbox{defbox}{colback=black!5!white, colframe=black!75!white, boxrule=0mm, leftrule=1mm, sharp corners, breakable}

\theoremstyle{plain}
\theoremstyle{definition}

\newtheorem{deef}{Definition}[section]

\theoremstyle{plain}
\newtheorem{lemm}{Lemma}[section]
\newtheorem*{thrm}{Theorem}
\newtheorem{prop}[lemm]{Proposition}
\newtheorem{corr}[lemm]{Corollary}

\theoremstyle{definition}
\newtheorem{exxx}{Example}[section]
\newtheorem{def7}{Remark}[section]
\theoremstyle{remark}

\newtheoremstyle{exercise}
{}{}{}{}
{\sffamily\bfseries}
{.}
{ }{\thmname{#1}\thmnumber{\phantom{a}#2}\thmnote{\textnormal{\phantom{i}(#3)}}}

\theoremstyle{exercise}

\numberwithin{exe}{section}

\theoremstyle{definition}

\usepackage{titlesec}
\titleformat*{\subsection}{\Large\sffamily\bfseries}
\titleformat*{\section}{\LARGE\sffamily\bfseries}
\titleformat*{\subsubsection}{\large\sffamily\bfseries}

\usepackage{textcomp}

\newcommand{\fk}[1]{\mathfrak{#1}}

\newcommand{\ii}{\mathrm{i}}

\newcommand{\ank}[1]{\left\langle #1 \right\rangle}
\newcommand{\bank}[1]{\big\langle #1 \big\rangle}

\newcommand{\sank}[1]{\langle #1 \rangle}

\newcommand{\wh}[1]{\widehat{#1}}

\DeclareMathOperator{\ttr}{tr}

\DeclareMathOperator{\supp}{supp}

\DeclareMathOperator{\End}{End}

\DeclareMathOperator{\detz}{det_{\zeta}}

\DeclareMathOperator{\conf}{Cfig}
\DeclareMathOperator{\ord}{ord}
\DeclareMathOperator{\sreg}{s}

\newcommand{\nrm}[1]{\left\lVert#1\right\rVert}

\newcommand{\mss}[1]{\mathscr{#1}}

\newcommand{\lto}{\longrightarrow}

\newcommand{\one}{\mathbbm{1}}

\newcommand{\bnrm}[1]{\big\lVert#1\big\rVert}

\renewcommand{\tilde}{\widetilde}
\renewcommand{\ge}{\geqslant}
\renewcommand{\le}{\leqslant}
\newcommand{\defeq}{\overset{\mathrm{def}}{=}}
\newcommand{\heueq}{\overset{\mathrm{heu}}{=}}

\usepackage{pdfpages}
\usepackage{blkarray}

\usepackage{fancyhdr}
\numberwithin{equation}{section}

\title{Entanglement Entropy and Cauchy-Hadamard Renormalization}
\author{Benoit Estienne\footnote{Sorbonne Universit\'e, email: \href{mailto:estienne@lpthe.jussieu.fr}{estienne@lpthe.jussieu.fr}} ~and Jiasheng Lin\footnote{Sorbonne Universit\'e, email: \href{mailto:jiasheng.lin@imj-prg.fr}{jiasheng.lin@imj-prg.fr}}}
\date{}

\begin{document}

\maketitle

\begin{abstract}
This note presents a purely geometric construction of the so-called twist-field correlation functions in Conformal Field Theory (CFT), derived from \textit{conical singularities}. This approach provides a purely mathematical interpretation of the seminal results in physics by Cardy and Calabrese on the entanglement entropy of quantum systems. Specifically, we begin by defining CFT partition functions on surfaces with conical singularities, using a ``Cauchy-Hadamard renormalization'' of the Polyakov anomaly integral. Next, we demonstrate that for a branched cover $f:\Sigma_d\lto \Sigma$ with $d$ sheets, where the cover inherits the pullback of a smooth metric from the base, a specific ratio of partition functions on the cover to the base transforms under conformal changes of the base metric in the same way as a \textit{correlation function of CFT primary fields} with specific conformal weights. We also provide a discussion of the physical background and motivation for entanglement entropy, focusing on path integrals and the replica trick, which serves as an introduction to these ideas for a mathematical audience.

\end{abstract}

\tableofcontents

\section{Introduction}

An (Euclidean) Conformal Field Theory (CFT) is a theory which, in one way, emerges as the \textit{scaling limit} of a lattice statistical mechanics model with \textit{critical parameters} (at phase transition), as the lattice spacing goes to zero. Due to its conformal covariance properties, such a theory can, in two dimensions, naturally be defined on a surface with its conformal structure ($=$ complex structure for oriented surfaces). 
Broadly speaking and for the purposes of this paper,  the data of such a theory consists of assigning to each (closed) Riemann surface~$\Sigma$ with metric~$g$ a number~$\mathcal{Z}(\Sigma,g)$ (the partition function) and a collection of correlation functions defined on~$\Sigma^n_{\ne}$ (subspace of the~$n$-fold product~$\Sigma^n$ with non-coincident points), which satisfy well-defined transformation rules under conformal changes of the metric~$g$ (see definition \ref{def-cft}). In comparison to the original lattice problem, the number~$\mathcal{Z}(\Sigma,g)$ corresponds to the normalization constant for the Gibbs measure, while the functions on~$\Sigma^n_{\ne}$ describe the scaling limit of the probabilistic correlation functions of~$n$ local observables of the spin configurations under the Gibbs measure. 

It has long been recognized in the physics literature that the presence of \textsf{conical singularities} in the underlying metric alters the conformal covariance properties of the partition function $\mathcal{Z}(\Sigma,g)$ in a way that mimics the behavior of correlation functions.  This paper aims to provide a rigorous framework for these observations, with a particular emphasis on partition functions on branched coverings, motivated by considerations of entanglement entropy (which we introduce in Section \ref{sec-ent-ent-main}). We remark that a CFT in the sense of the previous paragraph makes sense, \textit{a priori}, only for smooth metrics and smooth conformal changes. The first task of this paper thus lies in defining in a simple and natural manner the number~$\mathcal{Z}(\Sigma,\tilde{g})$ when the metric~$\tilde{g}$ admits finitely many isolated \textsf{conical singularities}. Then, applying this definition to metrics coming from pull-backs of ramified holomorphic maps, we observe that certain ratios of these numbers behave exactly like what is expected of correlation functions on the target surface where the variables are the critical values. More precisely, the main result is the following.

\begin{thrm}
  [proposition \ref{prop-main-app-ram-cov}] \footnote{See also ``Riemann surface terminologies'' on the next page. Moreover, we will assume that all Riemann surfaces in this paper are connected, unless otherwise specified.} Let~$\Sigma_d$,~$\Sigma$ be closed Riemann surfaces, with a smooth conformal metric~$g$ on~$\Sigma$, and~$f:\Sigma_d\lto \Sigma$ a ramified $d$-sheeted holomorphic map, whose critical values are~$w_1$, \dots,~$w_p$. Consider a conformal field theory with central charge~$c$ whose partition function is denoted~$\mathcal{Z}$. Pick~$h\in C^{\infty}(\Sigma)$ on~$\Sigma$. Then under the definition \ref{def-renom-part-func} for $\mathcal{Z}(\Sigma_d, f^* e^{2h} g)$ and $\mathcal{Z}(\Sigma_d, f^* g)$ we have
	  \begin{equation}
	    \frac{\mathcal{Z}(\Sigma_d, f^* e^{2h} g)}{\mathcal{Z}(\Sigma, e^{2h} g)^d} = e^{- \sum_j  h(w_j)\Delta_{j}} \frac{\mathcal{Z}(\Sigma_d, f^* g)}{\mathcal{Z}(\Sigma, g)^d}, 
	    \label{eqn-intro-conf-weights}
	  \end{equation}
	  where
	  \begin{equation}
	    \Delta_{j} \defeq \frac{c}{12}  \sum_{z \in f^{-1}(w_j)}\Big(\ord_f(z) - \frac{1}{\ord_f(z)} \Big)
	    \label{}
	  \end{equation}
	  are the \textsf{conformal weights} or \textsf{scaling dimensions}. Here $\ord_f(z)$ denotes the \textsf{order} or \textsf{multiplicity} of $f$ at $z\in \Sigma_d$.
\end{thrm}

The relation between CFT and conical singularities has been explored in various physical contexts. Early investigations, particularly within string theory, focused on the construction of orbifold CFTs \cite{Dixon}. In this framework, Knizhnik \cite{Knizhnik} examined the behavior of CFTs defined on flat branched coverings of $\mathbb{CP}^1$. More generally, the emergence of universal logarithmic divergences in the free energy due to conical singularities was recognized by Cardy and Peschel \cite{Cardy_Peschel}. Motivated by both black hole entropy and quantum information theory, entanglement entropy in CFTs was explored in seminal works \cite{Holzhey} and \cite{Cardy_Calabrese}, where it was shown that the universal contributions from conical singularities play a crucial role in understanding the scaling behavior of entanglement entropy in one-dimensional quantum critical systems. In this context, the Rényi entropy was identified as the free energy on $d$-sheeted branched covering of certain flat Riemann surface (typically a cylinder or a torus). Further studies, such as in \cite{Wiegmann}, extended the analysis of conical singularities to hyperbolic surfaces, with applications in the context of the quantum Hall effect \cite{KMMW}. 

Conical singularities is also a relatively well-studied subject mathematically. They are one of the simplest types of singularities that can appear on a Riemann surface. The main motivations include spectral geometry (hearing the shape of a drum) \cite{Cheeger} and the Berger-Nirenberg problem of finding metrics with prescribed curvature \cite{HT, Troy}. Motivated by its higher dimensional analogue in complex geometry called K\"ahler-Einstein edge metrics, people have also studied the Ricci flow on surfaces with such singularities \cite{MRS}. Particularly relevant to the present work are recent investigations on the~$\zeta$-determinant of the Laplacian under these conical metrics and the Polyakov formulas \cite{AKR, Kalvin}, on which we shall make a more detailed comment in remarks \ref{rem-rel-kal} and \ref{rem-rel-akr}. We would also like to point out the interesting recent work \cite{JV} on Coulomb gas and the Grunsky operator where conical singularities (more precisely, ``corners'' on the boundary) also play a role. For more literature from these various perspectives we refer to the introductions of \cite{AKR, Kalvin}, section 1.2 of \cite{JV}, and section 2.E of \cite{MRS}. Finally, one precise relation between conical singularities and the probabilistically constructed Liouville CFT has been established in \cite{HRV}.

One important feature of the present work lies in its \textit{simplicity} and \textit{naturalness}. Essentially, the method involves only a close look at the geometry near the cone points and integration by parts (Green-Stokes formula). Consequently, the geometric meaning of each term that shows up in the result is transparent (see also remark \ref{rem-extra-quad-term}). Moreover, we only need rather weak regularity to be imposed on the regular metric potential at the cone points compared to other related works in the literature (e.g.\ to \cite{Kalvin} definition 2.1), and this is basically the assumption adopted in \cite{Troy} (see also remark \ref{rem-reg-met-pot-cone}). 

\paragraph{Organization.} In section \ref{sec-def} we define the three main objects dealt with by this paper: the CFT correlation functions, conical singularities, and the renormalized Polyakov anomaly; these are accompanied by a few pivotal lemmas. Section \ref{sec-ent-ent-main} explains the physical motivation which leads us to (\ref{eqn-intro-conf-weights}). In these sections (and hence the whole paper) no prior knowledge of QFT, CFT or statistical mechanics is assumed, as this paper serves also to introduce these ideas to the mathematics community. Only some notions of quantum mechanics are required to make sense of entanglement entropy. Then in section \ref{sec-geo-lemm}, we collect a few known facts about conical metrics (subsection \ref{sec-conic-remark}), prove a crucial scaling lemma (subsection \ref{sec-geo-scaling}), and compute the asymptotics of several logarithmically divergent integrals using integration by parts (subsection \ref{sec-int-by-part}). They are important as pointed out by remarks \ref{rem-well-def-renorm-anom}, \ref{rem-const-def-part-func}, and also in obtaining the final result. In section \ref{sec-renorm-tech} we tie up some loose ends around the definitions of the renormalized anomaly and partition function for conical metrics. In subsection \ref{sec-main-proof} we prove the main result and apply it to entanglement entropies, and in \ref{sec-literature} we comment on the two closely related work \cite{AKR, Kalvin}. Finally in appendix \ref{sec-app-poin-lelong} we recall a few things around the so-called ``Poincar\'e-Lelong lemma'' concerning the Laplacian of the log of the distance function on a Riemannian surface.

\paragraph{Future work.} In the present version of the work we have focused on the simplest case dealing only with partition functions to illustrate our (already simple) methods. It is clear that the same arguments could apply with minor modifications to obtain analogous results for Segal's amplitudes on surfaces with boundary (see \cite{Gaw} section 2.6) when the metric is ``flat-at-the-boundary'', and for the case of general boundaries with possible presence of ``corners'' (i.e.\ ``polygons'') by including the boundary term in the anomaly, as considered for example by \cite{AKR}. These considerations may be included in a future update of the present manuscript. However, another main aim of that work will be to present an equivalent rigorous construction of the quantities~$\mathcal{Z}(\Sigma_d,f^*g)/\mathcal{Z}(\Sigma,g)^d$ in (\ref{eqn-intro-conf-weights})  for the GFF by constructing the so-called ``twist fields'' in the physics literature. This is related to a singular version of the ``twisted Laplacians'' used e.g.\ in Phillips and Sarnak \cite{PSa}. Lastly, less apparent but interesting future investigations include obtaining the precise relation between this work and \cite{Kalvin} in light of Segal's gluing axioms, and exploring further connections with the Quantum Hall Effect and Coulomb gas in a rigorous manner.

\paragraph{Riemann surface terminology and asymptotic notations.} Let~$f:\Sigma'\lto \Sigma$ be a holomorphic map of closed (connected\footnote{We will assume that all Riemann surfaces in this paper are connected, unless otherwise specified.}) Riemann surfaces. If at~$z\in \Sigma'$ we have~$\dd f|_z=0$ we say that~$z$ is a \textsf{critical point} of~$f$ and~$w=f(z)$ is a \textsf{critical value} of~$f$. 
If in local holomorphic charts around~$z$ and~$w$ respectively~$f$ looks like a~$k$-th power, then we say~$k$ is the \textsf{order} of~$f$ at~$z\in \Sigma'$ (sometimes referred to as the \textsf{local degree} or \textsf{multiplicity}),  denoted~$k=:\ord_f(z)$. 
We say that~$f$ is a \textsf{ramified} (or \textsf{branched}) holomorphic map when we want to emphasize that it is not a strict covering map.
The number~$\sum_{z\in f^{-1}(w)}\ord_f(z)=:d$, which is same for every~$w\in \Sigma$, is called the \textsf{number of sheets} of~$f$ and we say~$f$ is~\textsf{$d$-sheeted}. Finally, the notation $u\asymp v$ under some limit process means $u/v\to C>0$ in that limit, and we write $u\sim v$ for the case $C=1$.

\paragraph{Acknowledgement.} J.L.\ thanks first of all his thesis advisor Nguyen-Viet Dang for raising attention to the cut-off method for treating the anomaly discussed in \cite{EKZ}, stressing its simplicity and flexibility, which initiated the present work, and various other discussions. He thanks Tat-Dat T\^o for pointing out the reference \cite{MRS} and some related considerations in the geometry community. He thanks Phan Th\`anh Nam for pointing out some information about the partial trace. 
He also thanks Yuxin Ge, Colin Guillarmou, Semyon Klevtsov and Eveliina Peltola for helpful discussions. Finally he thanks the doctoral school ED386 of Sorbonne Universit\'e and his laboratory IMJ-PRG for supporting his research. B.E. thanks Laurent Charles, Nguyen-Viet Dang, Colin Guillarmou and Tat-Dat T\^o for helpful discussions.

  \section{Definition of Main Objects}\label{sec-def}

 \subsection{Conformal Field Theory and Correlation Functions}

\begin{deef}\label{def-cft}
  In this paper, by a \textsf{2d Conformal Field Theory} we mean a rule that associates to each compact Riemannian surface~$\Sigma$ with metric~$g$ (a priori smooth) a complex number~$\mathcal{Z}(\Sigma,g)$ called the \textsf{partition function}, and a family~$\{\ank{\phi_{\alpha_1}(\cdot)\cdots}_{\Sigma,g}\}$ of functions of finite tuples of non-coincident points on~$\Sigma$, called \textsf{correlation functions of primary fields} (labelled by the~$\alpha$'s), such that the following two conditions hold:
  \begin{enumerate}[(i)]
    \item \textsf{diffeomorphism invariance:} if~$\Psi:\Sigma'\lto \Sigma$ is a diffeomorphism of smooth surfaces, then
      \begin{align}
	\mathcal{Z}(\Sigma',\Psi^*g)&=\mathcal{Z}(\Sigma,g), \label{eqn-part-func-diffeo-inv}\\
	\ank{\phi_{\alpha_1}(\Psi(x_1))\cdots \phi_{\alpha_n}(\Psi(x_n))}_{\Sigma,g}&=\ank{\phi_{\alpha_1}(x_1)\cdots \phi_{\alpha_n}(x_n)}_{\Sigma',\Psi^*g},
	\label{eqn-main-diffeo-inv}
      \end{align}
      for any~$x_1$, \dots,~$x_n\in \Sigma$ non-coincident, and as well
    \item \textsf{local scale (conformal) covariance:} if~$\sigma\in C^{\infty}(\Sigma)$ then
      \begin{align}
	\mathcal{Z}(\Sigma,\me^{2\sigma}g)&=\exp\Big( \frac{c}{24\pi}\int_{\Sigma}^{}(|\nabla_g \sigma|_g^2+2K_g\cdot\sigma)\dd V_g \Big)\cdot \mathcal{Z}(\Sigma,g),\label{eqn-ord-polya-anom}\\
	\ank{\phi_{\alpha_1}(x_1)\cdots \phi_{\alpha_n}(x_n)}_{\Sigma,\me^{2\sigma}g}&=\prod_{j=1}^n \me^{-\sigma(x_j)\Delta_{\alpha_j}}\ank{\phi_{\alpha_1}(x_1)\cdots \phi_{\alpha_n}(x_n)}_{\Sigma,g},
	\label{eqn-main-conformal-cov}
      \end{align}
      where~$K_g$ denotes the Gauss curvature of~$g$ (half the scalar curvature), the constant~$c\in \mb{R}$ is called the \textsf{central charge}, charateristic of the specific theory at hand, and constants~$\Delta_{\alpha}\in\mb{R}$ the \textsf{conformal weights}, charateristic of the theory as well as the fields~$\phi_{\alpha}$. 
      \end{enumerate}
\end{deef}

\begin{deef}
  The quantity
  \begin{equation}
    A_{\Sigma}(\me^{2\sigma}g,g) \defeq  \frac{1}{24\pi}\int_{\Sigma}^{}(|\nabla_g \sigma|_g^2+2K_g\cdot\sigma)\,\dd V_g 
    \label{eqn-def-ord-poly-anomaly}
  \end{equation}
  for ~$\sigma\in C^{\infty}(\Sigma)$, that appears in the exponential in (\ref{eqn-ord-polya-anom}) is usually called the \textsf{Weyl} or \textsf{Polyakov Anomaly} of $\me^{2\sigma}g$ against $g$. 
\end{deef}

\begin{def7}\label{rem-ord-cocycle}
  The quantities~$A_{\Sigma}$ has the so-called \textsf{cocycle property}, that is, if~$g_3=\me^{2h}g_2$ and~$g_2=\me^{2\sigma}g_1$ where~$h$,~$\sigma\in C^{\infty}(\Sigma)$, then
  \begin{equation}
    A_{\Sigma}(g_3,g_2)+A_{\Sigma}(g_2,g_1)=A_{\Sigma}(g_3,g_1),
    \label{}
  \end{equation}
  as one can check using the relations (\ref{eqn-scale-vol}) --- (\ref{eqn-scale-lap}). One essential ingredient of this work is to define a ``renormalized'' version of the anomaly of a conically singular metric against a smooth one (definition \ref{def-renorm-polya-conic}). With these definitions the cocycle property will be modified accordingly when involving the conical metrics, see proposition \ref{prop-main-conical-scaling}. 
\end{def7}

\begin{exxx}[Gaussian Free Field] \label{exp-gff}
  Let~$(\Sigma,g)$ be a Riemannian surface with smooth metric~$g$ whose Laplacian is denoted~$\Delta_g$ (negative). Consider
  \begin{equation}
    \mathcal{Z}(\Sigma,g)\defeq \detz'(-\Delta_g)^{-\frac{1}{2}},
    \label{}
  \end{equation}
  where
  \begin{equation}
    \detz'(-\Delta_g)\defeq \exp\Big( -\frac{\dd}{\dd s}\Big|_{s=0}\sum_{j=1}^{\infty}\lambda_j^{-s} \Big),
    \label{}
  \end{equation}
  with~$0=\lambda_0<\lambda_1\le \lambda_2\le\cdots$ being the eigenvalues of~$-\Delta_g$ counted with multiplicity, called the~\textsf{$\zeta$-regularized determinant} \cite{RS}. Then~$\mathcal{Z}$ satisfies (\ref{eqn-ord-polya-anom}) (and (\ref{eqn-part-func-diffeo-inv}) trivially) with~$c=1$, by the \textsf{Polyakov formula} (see \cite{PW} appendix B). This~$\mathcal{Z}$ corresponds to the ``total mass'' of the formal measure~$\exp(-\frac{1}{2}\int_{\Sigma}^{}|\nabla_g \phi|^2 \dd V_g)\dd \mathcal{L}(\phi)$ on the space of zero-average distributions~$\mathcal{D}'_0(\Sigma)$ on~$\Sigma$, with~$\mathcal{L}$ being the non-existent Lebesgue measure there. In this case, an actual Gaussian probability measure~$\mu_{\mm{GFF}}^{\Sigma}=:\mu$ can indeed be defined on~$\mathcal{D}'_0(\Sigma)$ that corrresponds to~$\mathcal{Z}(\Sigma,g)^{-1}\exp(-\frac{1}{2}\int_{\Sigma}^{}|\nabla_g \phi|^2 \dd V_g)\dd \mathcal{L}(\phi)$, called the (massless) \textsf{Gaussian Free Field} \cite{PWer}. This is characterized by the formal covariance property (pretending that the point values $\phi(x)$, $x\in\Sigma$, are legitimate real random variables)
  \begin{equation}
    \mb{E}_{\mu}\big[\phi(x)\phi(y)\big]=G_{\Sigma}(x,y),\quad\textrm{and}\quad\mb{E}_{\mu}\big[\phi(x)\big]\equiv0,\quad\quad x,y\in \Sigma,
    \label{}
  \end{equation}
  with~$G_{\Sigma}$ being the Green function which is the integral kernel of~$(-\Delta_g)^{-1}P_{\ker \Delta_g}^{\perp}$. Now, after an appropriate renormalization process which we do not detail here (related to obtaining a finite value for~$G_{\Sigma}(x,x)$), denoted~``$\mathcal{R}$'', the quantities
  \begin{equation}
    \bank{\phi_{\alpha_1}(x_1)\cdots \phi_{\alpha_n}(x_n)}_{\Sigma,g}\defeq \textrm{``}\mathcal{R}\textrm{''}\mb{E}_{\mu}\big[\me^{\ii \alpha_1\phi(x_1) }\cdots \me^{\ii \alpha_n\phi(x_n) }\big],
    \label{}
  \end{equation}
  where now~$\alpha_j\in\mb{R}$, transform according to (\ref{eqn-main-conformal-cov}) with~$\Delta_{\alpha_j}:=\alpha_j^2/4\pi$.
\end{exxx}

The rules (i), (ii) actually determine, say, the two-point function up to a constant for some simple geometries.

\begin{lemm}\label{lemm-two-pt-func}
  Consider a CFT defined on the Riemann sphere~$\mb{S}^2=\mb{C}\cup\{\infty\}$, equipped with the Fubini-Study metric~$g_{\mm{FS}}(z):=4(1+|z|^2)^{-2}|\dd z|^2$. Take two primary fields~$\phi_1$,~$\phi_2$ with conformal weights~$\Delta_1$,~$\Delta_2$. Then
  \begin{equation}
  \bank{\phi_1(u)\phi_2(v)}_{\mb{S}^2,g_{\mm{FS}}}=
  \left\{
  \begin{array}{ll}
     C\sin\big( \frac{1}{2}d_{\mm{FS}}(u,v) \big)^{-2\Delta},&\textrm{when }\Delta_1=\Delta_2=\Delta,\\
    0 &\textrm{otherwise,}
  \end{array}
  \right.
  \label{}
\end{equation}
where $u\ne v$ and~$C$ is a non-zero constant that cannot be determined from the rules (i) and (ii) alone.
\end{lemm}

\begin{proof}
Without loss of generality we suppose $u$, $v\ne \infty$. Given a M\"obius transformation~$\psi\in \mm{PSL}(2,\mb{C})$ that sends~$0\mapsto u$,~$\infty\mapsto v$, we have by (\ref{eqn-main-diffeo-inv}) and (\ref{eqn-main-conformal-cov}),
	\begin{align}
	  \bank{\phi_1(u)\phi_2(v)}_{\mb{S}^2,g_{\mm{FS}}}
	  &=\bank{\phi_1(0)\phi_2(\infty)}_{\mb{S}^2,\psi^* g_{\mm{FS}}}\nonumber\\
	  &=\me^{-\Delta_1 \sigma(0)-\Delta_2\sigma(\infty)}\bank{\phi_1(0)\phi_2(\infty)}_{\mb{S}^2,g_{\mm{FS}}},\label{eqn-two-point-examp-calc}
	\end{align}
	where~$\psi^*g_{\mm{FS}}=\me^{2\sigma}g_{\mm{FS}}$. The above must hold for \textit{all} such M\"obius maps. These are of the form
	\begin{equation}
	\psi(z)=\frac{z v - z_0 u}{z - z_0}.
	\label{}
      \end{equation} 
      where $z_0  = \psi^{-1}(\infty) \in \mathbb{C}\setminus \{0 \}$ parametrizes the preimage of infinity. 
      This gives
      \begin{equation}
	(\psi^* g_{\mm{FS}})(z)=\frac{4|\psi'(z)|^2}{(1+|\psi(z)|^2)^2}|\dd z|^2=\frac{(1+|z|^2)^2}{(1+|\psi(z)|^2)^2}\frac{|v-u|^2 |z_0|^2}{|z - z_0|^4} g_{\mm{FS}}(z).
	  \label{}
	\end{equation}
 Therefore, according to (\ref{eqn-two-point-examp-calc}),
\begin{equation}
  \bank{\phi_1(u)\phi_2(v)}_{\mb{S}^2,g_{\mm{FS}}}=C\Big(\frac{1+|u|^2}{|u -v|}\Big)^{\Delta_1} \Big(\frac{1+|v|^2}{|u -v|}\Big)^{\Delta_2} |z_0|^{\Delta_2- \Delta_1}
\end{equation}
where $C =: \ank{\phi_1(0)\phi_2(\infty)}_{\mb{S}^2,g_{\mm{FS}}}$.
Since the correlation function should not depend on $z_0$, one finds 
\begin{equation}
  \textrm{either } \Delta_1 = \Delta_2 \quad \textrm{or} \quad C=0.
  \label{}
\end{equation}
Summing up,
\begin{equation}
  \bank{\phi_1(u)\phi_2(v)}_{\mb{S}^2,g_{\mm{FS}}}=
  \left\{
  \begin{array}{ll}
    C \left(1+|u|^2\right)^{\Delta} \left(1+|v|^2\right)^{\Delta}  |u - v|^{-2 \Delta}
    ,&\textrm{when }\Delta_1=\Delta_2=\Delta,\\
    0 &\textrm{otherwise.}
  \end{array}
  \right.
  \label{}
\end{equation}
where the constant $C  \in \mathbb{C}$ is $C= \bank{\phi_1(0)\phi_2(\infty)}_{\mb{S}^2,g_{\mm{FS}}}$.
We note here that the quantity
	\begin{equation}
	  d_{\mm{ch}}(u,v)\defeq \frac{2|u-v|}{\sqrt{(1+|u|^2)(1+|v|^2)}}
	  \label{}
	\end{equation}
	is the so-called \textsf{spherical chordal distance} between~$u$ and~$v$ (see \cite{Ahlfors} page 20), and is related to the actual spherical distance~$d_{\mm{FS}}(u,v)$ by
	\begin{equation}
	  d_{\mm{ch}}(u,v)=2\sin\Big( \frac{d_{\mm{FS}}(u,v)}{2} \Big).
	  \label{}
	\end{equation}
 This gives us the result.
\end{proof}

      \subsection{Conical Metrics}
\label{sec-singular-conic}

In this section we describe precisely our geometric set-up by adopting some of the terminologies of Troyanov \cite{Troy}.
	\begin{deef}
	  Let~$\Sigma$ be a closed Riemann surface. A \textsf{generalized conformal metric} on~$\Sigma$ is a distributional Riemannian metric~$\tilde{g}$ on~$\Sigma$ such that for any local complex coordinate~$z_U:U\lto \mb{C}$ defined on~$U\subset \Sigma$ we have
	  \begin{equation}
	    \tilde{g}=\varrho_U(z_U)\,|\dd z_U|^2
	    \label{}
	  \end{equation}
	  for some positive measurable function~$\varrho_U$ on~$U$. We say~$\tilde{g}$ is a \textsf{smooth conformal metric} if the functions~$\rho_U$ are all smooth and $0<c_{U,z}\le \varrho_U\le C_{U,z}$, for some constants $c_{U,z}$, $C_{U,z}$ depending on $U$ and the coordinate.
	\end{deef}

 Smooth conformal metrics exist on any Riemann surface by simple constructions using partitions of unity.

	\begin{deef}
	  A \textsf{real divisor} on a closed Riemann surface~$\Sigma$ is a finite formal sum
	  \begin{equation}
	    D=\sum_{j=1}^p \gamma_j z_j,\quad\quad \gamma_j\in\mb{R},~z_j\in\Sigma,~z_i\ne z_j\textrm{ for }i\ne j.
	    \label{}
	  \end{equation}
	  The number~$|D|:=\sum_j \gamma_j$ is the \textsf{degree} of~$D$. The set $\supp D:=\{z_j~|~\gamma_j\ne 0\}$ is the \textsf{support} of $D$.
	\end{deef}

    	\begin{deef}\label{def-singularities}
	  Let~$\Sigma$ be a closed Riemann surface. We say that a generalized conformal metric~$\tilde{g}$ has an (admissible, isolated) \textsf{conical singularity} (or \textsf{cone point}) of \textsf{order}~$d>0$ or \textsf{exponent} $\gamma:=d-1>-1$ at~$z_0\in\Sigma$, if~$z_0$ has a neighborhood~$U\subset \Sigma$, such that
\begin{equation}
  \tilde{g}|_{U\setminus z_0} =d_g(\bullet,z_0)^{2\gamma}\cdot \me^{2\varphi_{U,g}}g,
  \label{eqn-conic-sing-local-form}
\end{equation}
for some smooth conformal metric~$g$ on~$\Sigma$ and some function~$\varphi_{U,g}$ on~$U$, called the \textsf{(regular) metric potential} (against $g$), which is continuous on $U$, smooth on $U\setminus z_0$, and $\Delta_g \varphi_{U,g}\in L^1(U,g)$. Moreover, we require that the Gauss curvature of~$\tilde{g}$ on~$U\setminus z_0$ is bounded.
	\end{deef}

    \begin{def7}
  If~$\gamma=0$ (resp.\ $d=1$) then the point is called \textsf{regular}. In a neighborhood~$U$ of a regular point, the regularity of the metric potential~$\varphi_{U,g}$ is closely related to that of the Gauss curvature (given by (\ref{eqn-liouville-eqn}) as an~$L^1(U,g)$ function) of~$\tilde{g}$ on~$U$. For example, if the Gauss curvature is smooth then one could deduce that~$\varphi_{U,g}$ is also smooth. For details see \cite{Troy2} proposition 1.2. Our results will be consistent with the presence of regular points upon setting the corresponding exponents to zero, as one could check. 
\end{def7}

	\begin{deef}
	  Let~$\Sigma$ be a closed Riemann surface,~$D=\sum_{j=1}^p \gamma_j z_j$ be a real divisor with $\gamma_j>-1$. We say that a generalized conformal metric~$\tilde{g}$ \textsf{represents} the divisor~$D$ if it has an isolated conical singularity of exponent $\gamma_j$ at $z_j$ respectively for each $j$, and smooth on $\Sigma\setminus \supp D$.
	\end{deef}

\begin{def7}
	  Our presentation differs from what usually happens in the literature where the background metric~$g$ is the local flat one coming from some compatible complex coordinate on the neighborhood~$U$. We are faced with the natural question of whether~$\tilde{g}$ would have the same form (\ref{eqn-conic-sing-local-form}) with the same regularity on the metric potential if we switched to another smooth conformal background metric~$g_1=\me^{2h}g$,~$h\in C^{\infty}(\Sigma)$. Indeed, in this case
	  \begin{equation}
	    d_g(\bullet,z_0)^{2\gamma}\cdot \me^{2(\varphi_{U,g}-h)}g_1=d_{g_1}(\bullet,z_0)^{2\gamma}\Big( \frac{d_{g}(\bullet,z_0)}{d_{g_1}(\bullet,z_0)} \Big)^{2\gamma}\cdot \me^{2(\varphi_{U,g}-h)}g_1,
	    \label{}
	  \end{equation}
	  and the ratio of distances~$d_{g}(\bullet,z_0)/d_{g_1}(\bullet,z_0)$ is continuous on~$U$ with a limit~$\me^{-h(z_0)}$ at~$z_0$, smooth on~$U\setminus z_0$ and~$\Delta_{g_1}\log d_g -\Delta_{g_1}\log d_{g_1}$ is in $L^1(U,g_1)$ as a distribution by lemma \ref{lemm-bound-lap-log-ratio}. In the traditional setting if both~$g$ and~$g_1$ are locally~$|\dd z|^2$ and~$|\dd w|^2$ for some complex coordinates~$z$ and~$w$, then the regular metric potentials differ by a harmonic function, as is well known.
In particular, if we do test against such a coordinate metric then our conditions on the regular metric potential satisfy what Troyanov calls that of an \textsf{admissible metric} in \cite{Troy2}. 
\end{def7}

The following lemma describes the source of the conical singularities that the main application of the main result of this article tries to target.

\begin{lemm}
     Let~$\Sigma'$,~$\Sigma$ be closed Riemann surfaces and~$f:\Sigma'\lto \Sigma$ a holomorphic map. Let~$z_0\in \Sigma'$ be a critical point of order~$k$ and~$w_0=f(z_0)\in \Sigma$ the corresponding critical value. Let~$g$ be a smooth conformal metric on~$\Sigma$. Then~$f^*g$ has a conical singularity at~$z_0\in \Sigma'$ with order~$k$.
  \end{lemm}

 \begin{proof}
    There exists holomorphic charts $(U,z)$ and $(V,w)$ around $z_0$ and $w_0 = f(z_0)$ in which $f$ is represented  by $z\mapsto z^k$. The metric $g$ is locally of the form $g\big|_{V} = e^{2 h} |\dd w|^2$ for some $h \in C^{\infty}(V)$, and 
     \begin{equation}
f^*g\big|_{U} =  e^{2 f^*h} \,  |\dd f|^2 = e^{2 f^*h} \, k^2  |z|^{2(k-1)} |\dd z|^2  = d_{g'}(\bullet,z_0)^{2(k-1)} \, k^2e^{2 f^* h}  \, g'\big|_{U}
 \label{eqn-proof-pull-back-met-con}
    \end{equation}  
where $g'$ is any smooth conformal metric on $\Sigma'$ whose restriction to $U$ is $g'\big|_{U} = |\dd z|^2$ (such a metric can be manufactured using smooth bump functions, and taking $U$ smaller if necessary). Now (\ref{eqn-proof-pull-back-met-con})  is of the form (\ref{eqn-conic-sing-local-form}) with metric potential~$\varphi= f^* h + \log k$. Note that $\varphi$ is actually smooth across~$U$ in this situation.
  \end{proof}

      \subsection{Renormalized Anomaly}\label{sec-def-renorm-anom-part}

\begin{deef}\label{def-main-renorm-anomaly}
    Let~$\Sigma$ be a closed Riemann surface and~$\tilde{g}$ a generalized conformal metric representing~$D=\sum_{j=1}^p \gamma_j z_j$ with~$\gamma_j>-1$. Suppose~$g$ is a smooth conformal metric on~$\Sigma$ and~$\tilde{g}=\me^{2\sigma}g$ for some~$\sigma\in C^{\infty}(\Sigma\setminus \supp D)$. We define the \textsf{renormalized Polyakov anomaly} against metric~$g$ to be
  \begin{equation}
   \mathcal{R}A_{\Sigma}(\tilde{g},g)\defeq \frac{1}{24\pi}\lim_{\varepsilon\to 0^+}\Big[ 
      \int_{\Sigma\setminus \bigcup_{i=1}^p B_{\varepsilon}(z_i,\tilde{g})}(|\nabla_g \sigma|_g^2+2K_g\sigma)\dd V_g +2\pi\sum_{i=1}^p \frac{\gamma_i^2}{1+\gamma_i}\log(\varepsilon)
    \Big],
    \label{eqn-def-renorm-polya}
  \end{equation}
where~$B_{\varepsilon}(z_i,\tilde{g})$ is the metric disk of radius~$\varepsilon$ centered at~$z_i$ under the metric~$\tilde{g}$.
\end{deef} 

\begin{def7} \label{rem-well-def-renorm-anom} To show that $\mathcal{R}A_{\Sigma}(\tilde{g},g)<\infty$ we first show that it is true for~$g$ locally equal to~$|\dd w_j|^2$ near each~$z_j$ with the coordinate~$w_j$ coming from lemma \ref{lemm-reg-metric-pot} (Troyanov form). In this case the result follows from corollary \ref{cor-quad-blow-tilde}. Then we use lemma \ref{lemm-consist} with~$g_0$ being the special Troyanov form and~$g_1$ generic, to extend to the case of a generic smooth~$g$.
\end{def7}

\begin{def7}
    This method of renormalization dates back to Cauchy and Hadamard as they defined ``principal values'' of divergent integrals.
  \end{def7}

	\begin{deef}\label{def-renom-part-func}
	  Consider a conformal field theory with central charge~$c\in\mb{R}$ on the Riemann surface~$\Sigma$ and let~$\tilde{g}$ be a generalized conformal metric representing~$D=\sum_{j=1}^p \gamma_j z_j$ with~$\gamma_j>-1$. We define the \textsf{renormalized partition function} of $(\Sigma,\tilde{g})$ to be
	  \begin{equation}
	    \mathcal{Z}(\Sigma,\tilde{g})\defeq \exp\big(c\mathcal{R}A_{\Sigma}(\tilde{g},g)\big)\mathcal{Z}(\Sigma,g),
	    \label{eqn-def-renom-part-func}
	  \end{equation}
	  where~$g$ is any smooth conformal metric on~$\Sigma$ such that~$\tilde{g}=\me^{2\sigma}g$ for some~$\sigma\in C^{\infty}(\Sigma\setminus \supp D)$. 
	\end{deef}

\begin{def7}\label{rem-const-def-part-func} Lemma \ref{lemm-consist} ensures that $\mathcal{Z}(\Sigma,\tilde{g})$ is independent from the choice of reference metric $g$.
\end{def7}

      \section{Physical Interpretation}\label{sec-ent-ent-main}

In this section we explain in detail the so-called \textsf{entanglement entropies} of quantum systems and a technique for computing them using \textsf{path integrals}, called the \textsf{replica trick} in the physics community \cite{Cardy_Calabrese, Hea, RT, Witten}. This leads us to consider precisely the ratios of the form shown in (\ref{eqn-intro-conf-weights}) and to expect that they behave like correlation functions. We remark that what we present here is more of a conjectural, heursitic framework than rigorous proofs.

\subsection{Partial Trace and Entanglement}
In this subsection we discuss some generalities on describing entanglement between quantum systems. We start with a generic quantum system associated to a Hilbert space~$\mathcal{H}$. Recall that a \textsf{(quantum) statistical ensemble} (or \textsf{statistical state}) on~$\mathcal{H}$ is represented by a nonnegative trace class operator~$\rho$ on~$\mathcal{H}$ with~$\ttr_{\mathcal{H}}(\rho)=1$. This is usually called a \textsf{density operator}. Denote by~$\mathcal{J}_1(\mathcal{H})$,~$\mathcal{J}_{\infty}(\mathcal{H})$ and~$\mathcal{L}(\mathcal{H})$ the trace class, compact and bounded operators on~$\mathcal{H}$.

Now suppose the quantum system can be decomposed into two subsystems~$A$ and~$B$. In other words, assume that the Hilbert space~$\mathcal{H}$ is a tensor product:
      \begin{equation}
	\mathcal{H}=\mathcal{H}_A\otimes \mathcal{H}_B.
	\label{eqn-hilb-tens-decomp-gen}
      \end{equation}
      In order to describe the relation of the subsystems to the whole system, a useful operation is called the \textsf{partial trace}. This corresponds to taking trace ``over one component'' and one is left with an operator acting on the ``remainder component''. It could be defined rigorously as follows.
For each~$\rho\in \mathcal{J}_1(\mathcal{H})$, consider the linear functional
    \begin{equation}
      C\longmapsto \ttr_{\mathcal{H}}(\rho(C\otimes \one_B))
      \label{}
    \end{equation}
    for~$C\in \mathcal{J}_{\infty}(\mathcal{H}_A)$. We have
    \begin{equation}
      \big|\ttr_{\mathcal{H}}\big(\rho(C\otimes \one_B)\big)\big|\le \bnrm{\rho}_{\mathcal{J}_1(\mathcal{H})}\bnrm{C\otimes \one_B}_{\mathcal{L}(\mathcal{H})}= \bnrm{\rho}_{\mathcal{J}_1(\mathcal{H})}\bnrm{C}_{\mathcal{L}(\mathcal{H}_A)}.
      \label{eqn-gen-part-trace-bound}
    \end{equation}
    Thus (\ref{eqn-gen-part-trace-bound}) defines a bounded linear functional on~$\mathcal{J}_{\infty}(\mathcal{H}_A)$ with norm~$\nrm{\rho}_{\mathcal{J}_1(\mathcal{H})}$. Since~$\mathcal{J}_1(\mathcal{H}_A)$ is the Schatten dual of the ideal~$\mathcal{J}_{\infty}(\mathcal{H}_A)$ we obtain a unique representative~$\ttr_B(\rho)\in \mathcal{J}_1(\mathcal{H}_A)$.

    \begin{deef}
	  Let the Hilbert space~$\mathcal{H}$ be a tensor product as in (\ref{eqn-hilb-tens-decomp-gen}). Then the \textsf{partial trace over~$\mathcal{H}_B$} is the linear map
	  \begin{equation}
	    \left.
	    \begin{array}{rcl}
	       \ttr_B:\mathcal{J}_1(\mathcal{H}) &\lto& \mathcal{J}_1(\mathcal{H}_A),\\
	       \rho &\longmapsto &\ttr_B(\rho),
	    \end{array}
	    \right.
	    \label{}
	  \end{equation}
	  such that~$\ttr_B(\rho)$ is the unique operator on~$\mathcal{H}_A$ satisfying
	  \begin{equation}
	    \ttr_{\mathcal{H}_A}\big(\ttr_B(\rho)C\big) =\ttr_{\mathcal{H}}\big(\rho(C\otimes \one_B)\big)
	    \label{eqn-part-trace-dual-def}
	  \end{equation}
	  for all bounded operators~$C\in \mathcal{L}(\mathcal{H}_A)$.
	\end{deef}

    	\begin{def7}\label{rem-property-partial-trace}
	  For more details see Simon \cite{Sim4} pages 269-270, and also \cite{Sim1} theorem 3.2. In particular, from (\ref{eqn-gen-part-trace-bound}) one can see that the partial trace is continuous with respect to the~$\mathcal{J}_1(\mathcal{H})$- and~$\mathcal{J}_1(\mathcal{H}_A)$-norms (trace norms), and by taking~$C=|\psi\rangle\langle \psi|$,~$\psi\in \mathcal{H}_A$, in (\ref{eqn-part-trace-dual-def}) one could also see that~$\ttr_B(\rho)$ is nonnegative if~$\rho$ is.
	\end{def7}
      
      \begin{def7}
  Restricted to rank-one operators, we have
  \begin{equation}
    \ttr_B\big( |\varphi_A\otimes \varphi_B\rangle\langle \psi_A\otimes \psi_B| \big)=\langle \psi_B|\varphi_B\rangle |\varphi_A\rangle\langle \psi_A|,
    \label{}
  \end{equation}
  with~$|\varphi_A\rangle$,~$|\psi_A\rangle\in \mathcal{H}_A$,~$|\varphi_B\rangle$,~$|\psi_B\rangle\in \mathcal{H}_B$. Here the notation~$|w\rangle\langle v|$ for~$w$,~$v\in \mathcal{H}$ denotes the operator mapping any~$u\in \mathcal{H}$ to~$\ank{v,u}_{\mathcal{H}}w$. See also example \ref{exp-mat-part-trace}. 
\end{def7}

\begin{deef}
	Suppose~$\rho$ is a density operator over~$\mathcal{H}=\mathcal{H}_A\otimes \mathcal{H}_B$. Then the \textsf{reduced density operator over~$A$} is defined by
	\begin{equation}
	  \rho_A\defeq \ttr_B(\rho).
	  \label{}
	\end{equation}
      \end{deef}
It follows from remark \ref{rem-property-partial-trace} that $\rho_A$ is a density operator over~$\mathcal{H}_A$.
      \begin{deef}
Suppose~$\rho$ is a density operator over~$\mathcal{H}=\mathcal{H}_A\otimes \mathcal{H}_B$. We define
\begin{align}
  \textrm{\textsf{the entanglement entropy (over }}A\textrm{\textsf{)}} && \mathcal{S}_A&\defeq -\ttr_{\mathcal{H}_A}(\rho_A \log \rho_A), \\
  \textrm{\textsf{the }}n\textrm{\textsf{-th R\'enyi entropy (over }}A\textrm{\textsf{)}} && \mathcal{S}^{(n)}_A &\defeq \frac{1}{1-n}\log\ttr_{\mathcal{H}_A}(\rho_A^n),\quad \textrm{for }n>1.
  \label{eqn-def-renyi-ent}
\end{align}
      \end{deef}

\begin{def7}
	In general the quantities~$\mathcal{S}_A$ and~$\mathcal{S}_A^{(n)}$ defined above do not equal~$\mathcal{S}_B$ and~$\mathcal{S}_B^{(n)}$, the corresponding quantities defined the other way round by first taking partial trace over~$A$. However,~$\mathcal{S}_A=\mathcal{S}_B$ and~$\mathcal{S}_A^{(n)}=\mathcal{S}_B^{(n)}$ do hold if~$\rho$ is a \textsf{pure state}, namely a rank-one projector written as~$\rho=|\psi\rangle\langle \psi|$ for some~$|\psi\rangle\in \mathcal{H}$. Therefore, for pure states it makes sense to define \textit{the} entanglement and R\'enyi entropies~$\mathcal{S}:=\mathcal{S}_A=\mathcal{S}_B$ and~$\mathcal{S}^{(n)}:=\mathcal{S}_A^{(n)}=\mathcal{S}_B^{(n)}$. In the general case, one could then consider another quantity called the \textsf{mutual information}, defined as
	\begin{equation}
	  \mathcal{I}_{A,B}^{(n)}\defeq \mathcal{S}_A^{(n)}+\mathcal{S}_B^{(n)}-\mathcal{S}_{\mathcal{H}}^{(n)},
	  \label{}
	\end{equation}
	where~$\mathcal{S}_{\mathcal{H}}^{(n)}$ is defined the same way as (\ref{eqn-def-renyi-ent}) by just removing the subscript~$A$. 
      \end{def7}

\begin{def7}
	  If~$\rho$ is trace class, then~$\rho_A$ is trace class and the complex power~$\rho_A^z$ is well-defined for~$\fk{Re}(z)>1$. If moreover~$\mathcal{S}_A<\infty$, one then has the relation
	  \begin{equation}
	    \mathcal{S}_A=\lim_{z\to 1^+}\mathcal{S}_A^{(z)}=-\lim_{z\to 1^+}\frac{\partial}{\partial z}\ttr_{\mathcal{H}_A}(\rho_A^z).
	    \label{}
	  \end{equation}
	  However, for continuum QFTs the tensor product decomposition~$\mathcal{H}=\mathcal{H}_A\otimes \mathcal{H}_B$ is generally not well-defined, and the only approach available which produces reasonable answers seems to be employing the so-called \textsf{replica} interpretation to be explained in the next subsection, or the equivalent approach by \textsf{twist fields} (both avoid defining the tensor product). These concerns the quantities~$\ttr_{\mathcal{H}_A}(\rho_A^n)$ for positive integral~$n$ and is also what we will focus on for this paper. 
	\end{def7}

    \begin{exxx}\label{exp-mat-part-trace}
  Let~$\mathcal{H}=\mb{C}^2\otimes \mb{C}^2 =:\mathcal{H}_A\otimes \mathcal{H}_B$, in this order. Consider the standard basis~$e_1\otimes e_1$,~$e_1\otimes e_2$,~$e_2\otimes e_1$ and~$e_2\otimes e_2$ (in this order) of~$\mathcal{H}$ where~$e_1$,~$e_2$ are standard base vectors of~$\mb{C}^2$. Then for~$\rho\in \End(\mb{C}^2\otimes \mb{C}^2)$ written
  \begin{equation}
    \rho=\left(
    \begin{array}{cccc}
      a_{11} &a_{12} &a_{13} &a_{14} \\
      a_{21} &a_{22} &a_{23} &a_{24} \\ 
      a_{31} &a_{32} &a_{33} &a_{34} \\
      a_{41} &a_{42} &a_{43} &a_{44} 
    \end{array}
    \right),
    \label{}
  \end{equation}
  we have
  \begin{equation}
    \ttr_B( \rho)=\left(
    \begin{array}{cc}
      a_{11}+a_{22} & a_{13}+a_{24} \\
      a_{31}+a_{42} & a_{33}+a_{44}
    \end{array}
    \right),\quad\quad 
    \ttr_A( \rho)=\left(
    \begin{array}{cc}
      a_{11}+a_{33} & a_{12}+a_{34} \\
      a_{21}+a_{43} & a_{22}+a_{44}
    \end{array}
    \right).
    \label{}
  \end{equation}
  Also, if~$\rho=\rho_A\otimes \rho_B$ is a Kronecker product, we have~$\ttr_B(\rho)=\ttr(\rho_B)\rho_A$,~$\ttr_A(\rho)=\ttr(\rho_A)\rho_B$.
\end{exxx}

\subsection{Path Integrals and Replica}\label{sec-main-replica}
      Now let us describe the geometric framework where the quantities~$\mathcal{S}_A$ and~$\mathcal{S}_A^{(n)}$ get represented by \textit{path integrals}, as presented in \cite{Cardy_Calabrese}. Here we shall describe the pictures heuristically and \textbf{not attempt at any rigor}. For this paper we focus on~$1+1$ dimensional field theories where the space is either the real line or a circle with a specific perimeter (as a Riemannian manifold), denoted generally by~$X$, and space-time a 2-dimensional Riemannian surface with or without boundary denoted~$\Sigma$ (for example,~$\Sigma=X\times [0,T]$ or~$X\times \mb{R}$). Each field theory comes with specified \textsf{field configuration spaces} over space and space-time, denoted~$\conf(X)$ and~$\conf(\Sigma)$, as well as over their sub-regions. Typically~$\conf(X)=\mm{Map}(X,\mss{V})$, the space of \textit{maps} from~$X$ to a \textsf{spin value space}/\textsf{target space}~$\mss{V}$. Such configuration spaces should allow
    \begin{enumerate}[(a)]
      \item restriction of a configuration onto a subregion~$A\subset X$ (or~$\Sigma$), heuristically a map
	\begin{equation}
	  \left.
	  \begin{array}{rcl}
	     \conf(X) &\lto &\conf(A),\\
	     \phi &\longmapsto &\phi|_A.
	  \end{array}
	  \right.
	  \label{}
	\end{equation}
	In other words one is allowed to \textit{localize} the field;
      \item each configuration~$\phi\in \conf(X)$ to be recovered from the pair~$(\phi|_A,\phi|_{A^c})$ of its restrictions onto~$A$ and~$A^c=X\setminus A$, and moreover any pair of configurations on complementary subregions should combine into a global configuration; this is to say heuristically,
	\begin{equation}
	  \conf(X)\cong \conf(A)\times \conf(A^c).
	  \label{eqn-decomp-local-complemen}
	\end{equation}
    \end{enumerate}
    
    \begin{def7}
	  The above statements are rigorous in the case~$X=\Lambda$ is a discrete lattice, and~$\conf(\Lambda)=\mm{Map}(\Lambda,\mss{V})$. Here ``subregions'' correspond to subsets of lattice sites.
	\end{def7}

     \begin{def7}
    Suppose~$\Sigma=X\times [0,T]$. Then~$\conf(\Sigma)$ could be considered as the space of ``paths'' through which a configuration over~$X$ evolves across the time interval~$[0,T]$. Indeed, if we have taken~$\conf(\Sigma)$ to be~$\mm{Map}(\Sigma,\mss{V})$, then we would have simply~$\mm{Map}(X\times[0,T],\mss{V})\cong \mm{Map}([0,T],\mm{Map}(X,\mss{V}))$. The integral of (\ref{eqn-heu-gen-amp}) below is thus an integral over a space of ``paths''.
  \end{def7}

    We now describe heuristically the Hilbert space and the time evolution. For the Hilbert space one takes
	\begin{equation}
	  \mathcal{H}_X\defeq L^2(\conf(X),\mathcal{L}),
	  \label{}
	\end{equation}
	where~$\mathcal{L}$ denotes the non-existent \textit{Lebesgue measure} on the configuration space~$\conf(X)$. Under the path-integral formalism (Euclidean\footnote{See the first paragraph of subsection \ref{sec-dens-op-imag-time}.}) time evolution across~$\Sigma=X\times [0,T]$, namely over time~$T$, is represented by the integral operator
	\begin{equation}
	  \left.
	  \def\arraystretch{1.3}
	  \begin{array}{rcl}
	     U_T: \mathcal{H}_X &\lto &\mathcal{H}_X,\\
	     F&\longmapsto & (U_T F)(\psi)\defeq \ddp\int_{\conf(X)}^{}\mathcal{A}_T(\psi,\varphi)F(\varphi)\,\dd \mathcal{L}(\varphi),
	  \end{array}
	  \right.
	  \label{}
	\end{equation}
	with the integral kernel
	\begin{equation}
	  \mathcal{A}_T(\psi,\varphi)\defeq \int_{
	     \left\{  \phi\in \conf(\Sigma)~\middle|~ \substack{
	     \phi|_{X\times\{0\}}=\varphi,\\ 
	     \phi|_{X\times \{T\}}=\psi
	   }
      \right\}} \me^{-S_{\mm{EQFT}}(\phi)}\,\dd \mathcal{L}(\phi),
	  \label{eqn-heu-gen-amp}
	\end{equation}
	where now we integrate against the still non-existent Lebesgue measure on~$\conf(\Sigma)$ (with the indicated boundary conditions), and where~$S_{\mm{EQFT}}$ is the \textsf{action functional} (``E'' stands for ``Euclidean'') of the specific theory at hand.

    The above recipe for time evolution does not involve the fact that~$\Sigma$ is~$X\times [0,T]$, and indeed it makes sense for space-times~$\Sigma$ having any kind of geometry as long as~$\partial \Sigma=\partial_{\mm{ini}}\Sigma\sqcup \partial_{\mm{ter}}\Sigma$ and~$\partial_{\mm{ini}}\Sigma\cong \partial_{\mm{ter}}\Sigma\cong X$, namely its boundary has two components (called \textsf{initial} and \textsf{terminal}) both isometric to~$X$. Generalizing still further, Segal \cite{Segal} proposed a set of axioms that defines a QFT abstractly as a rule that associates Hilbert spaces to space manifolds and evolution operators to space-time manifolds that ``connects'' them (cobordisms). 
    
    Now denote by~$U_{\Sigma}$ and~$\mathcal{A}_{\Sigma}$ the evolution operator and its integral kernel corresponding to the space-time piece~$\Sigma$. Two important axioms of Segal (and Atiyah) are written as
	\begin{description}
	  \item[(composition)] if one has two space-time pieces~$\Sigma_1$,~$\Sigma_2$ and one glues them together by identifying~$\partial_{\mm{ter}}\Sigma_1$ with~$\partial_{\mm{ini}}\Sigma_2$ via~$X$, obtaining the piece~$\Sigma_2\circ \Sigma_1$, then we have
	    \begin{equation}
	      U_{\Sigma_2\circ \Sigma_1}=U_{\Sigma_2}\circ U_{\Sigma_1};
	      \label{}
	    \end{equation}
	  \item[(trace)] if one glues the space-time~$\Sigma$ with itself by identifying~$\partial_{\mm{ini}}\Sigma$ and~$\partial_{\mm{ter}}\Sigma$ via~$X$, obtaining~$\check{\Sigma}$, then
	    \begin{equation}
	      \ttr_{\mathcal{H}}(U_{\Sigma})=\mathcal{Z}(\check{\Sigma}),
	      \label{}
	    \end{equation}
	    where the number
	    \begin{equation}
	      \mathcal{Z}(\check{\Sigma})\defeq \int_{\conf(\check{\Sigma})}^{}\me^{-S_{\mm{EQFT}}(\phi)}\,\dd \mathcal{L}(\phi)
	      \label{eqn-heu-gen-part-func}
	    \end{equation}
	    is the \textsf{partition function} (same as what appears in definition \ref{def-cft}).
	\end{description}

\begin{def7}
	  Comparing (\ref{eqn-heu-gen-part-func}) with (\ref{eqn-heu-gen-amp}), one sees that the trace axiom corresponds formally to the fact that the ``trace'' of an integral operator is the integral along the diagonal of its kernel. For rigorous results concerning this statement see Simon \cite{Sim3} section 3.11.
	\end{def7}

    Now we try to incorporate considerations of subregions~$A\subset X$ into the framework described above. For simplicity we will take~$A$ to be a finite interval (remember~$X$ is either a circle or the real line). First of all, following (\ref{eqn-decomp-local-complemen}), we have formally
	\begin{equation}
	  \mathcal{H}_X \overset{\mm{heu}}{\cong} L^2\big(\conf(A)\times \conf(A^c),\mathcal{L}_{\conf(A)}\otimes \mathcal{L}_{\conf(A^c)}\big)\overset{\mm{heu}}{\cong} \mathcal{H}_{A}\otimes \mathcal{H}_{A^c}.
	  \label{}
	\end{equation}
	To include the partial trace, the trace axiom is now slightly extended by adding that taking the partial trace (over~$A^c$) corresponds to gluing ``partially'' (along~$A^c$ but not~$A$), leaving out a ``slit'' with two sides denoted~$A_-$ and~$A_+$. In terms of the integral kernels, if we denote by~$\mathcal{A}_{\Sigma}^A$ the kernel of~$\ttr_{A^c}(U_{\Sigma})$ acting on~$\mathcal{H}_{A}$, then we have
	\begin{equation}
	  \mathcal{A}_{\Sigma}^A(\psi_{A},\varphi_A)\heueq \int_{\conf(A^c)}^{} \mathcal{A}_{\Sigma}(\psi_A,\sigma_{A^c},\varphi_A,\sigma_{A^c})\,\dd \mathcal{L}(\sigma_{A^c}).
	  \label{}
	\end{equation}
    Here~$\ttr_{A^c}(U_{\Sigma})$ as an operator acting on~$\mathcal{H}_A$, the Hilbert space associated to an interval, is represented by a surface (space-time) with a ``slit'' (or \textsf{branch cut}/\textsf{defect line}) where the two ``sides'' of the slit are identified with two copies of the interval. Accordingly, we must also extend the composition axiom to incorporate this situation, namely to allow gluing of surfaces with slits along sides which represents composition of operators acting on Hilbert spaces over intervals. We shall assume that this has been done in the obvious manner.

    Finally we arrive at an interpretation of a quantity of the form~$\ttr_{\mathcal{H}_A}(\ttr_{A^c}(U_{\Sigma})^n)$ which appears in the expression for the R\'enyi entropy (\ref{eqn-def-renyi-ent}). Indeed, one starts with the surface~$\Sigma$ and glue its two ends ``partially'' along the parts corresponding to~$A^c$ in~$X$. We denote the resulting ``surface with slit'' by~$\check{\Sigma}\setminus A$. 
	Then the two sides of the slit gets identified with two copies of~$A$, which we denote by~$A_{\pm}$, that corresponds to approaching~$A$ from the two sides within~$\check{\Sigma}\setminus A$. Next, we take~$n$ copies of~$\check{\Sigma}\setminus A$ and glue them in a \textit{cyclic manner}, that is, we glue~$A_+^{(j)}$ on the~$j$-th copy to~$A_-^{(j+1)}$ on the~$(j+1)$-th copy,~$1\le j<n$, and finally~$A_+^{(n)}$ to~$A_-^{(1)}$. We denote the surface thus obtained by~$\check{\Sigma}_n$, called the $n$-th \textsf{replica}. 
    Importantly, $\check{\Sigma}_n$ comes equipped with a \textbf{metric} which is induced from the metric on the original space-time $\Sigma$. Equivalently, there is an obvious map~$f_n:\check{\Sigma}_n\lto \check{\Sigma}$, sending the end points of~$A$ to themselves, and any other point to its counterpart on the original copy. 
    Then~$f_n$ is a branched~$n$-sheeted cyclic covering with critical points being the two end points of the copies of~$A$, and the metric on~$\check{\Sigma}_n$ is the pull-back under~$f_n$ of the original metric on~$\check{\Sigma}$ (on~$\Sigma$). This will be a metric with \textit{conical singularities} at the two end points of $A$ which are not duplicated. If we assume the \textit{extended} versions of the composition and trace axioms discussed above, then we find
    \begin{equation}
	  \ttr_{\mathcal{H}_A}\big(\ttr_{A^c}(U_{\Sigma})^n\big) =\mathcal{Z}\big(\check{\Sigma}_n\big).
	  \label{}
	\end{equation}
The geometric correspondence is summarized in the table below.
    
      \begin{center}
  \def\arraystretch{1.3}
  \begin{tabular}[]{|c|c|}
    \hline
   Operation & Picture \\ \hline
    evolution operator $U_{\Sigma}$ & 
	 \raisebox{-0.5\height}{\includegraphics[width=5cm]{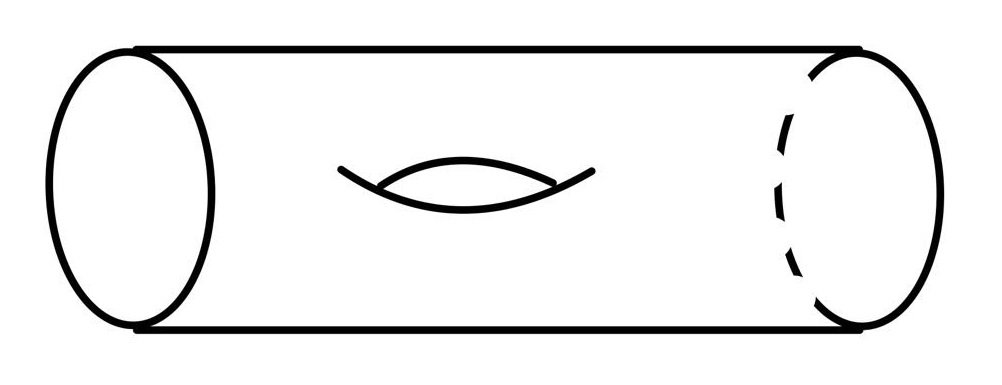}} 
    \\ \hline
	  $\ttr_{\mathcal{H}}(U_{\Sigma})$ & 
      \raisebox{-0.5\height}{\includegraphics[width=5cm]{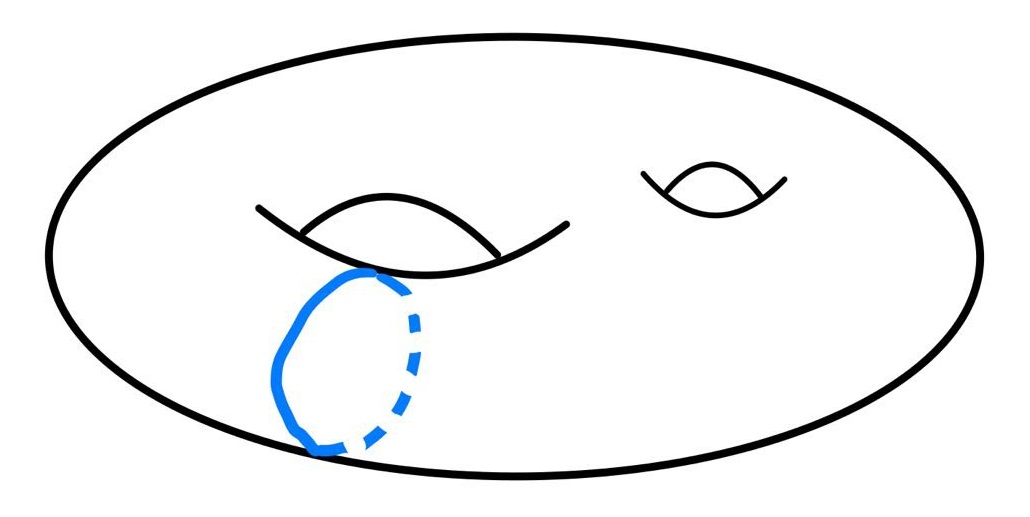}}
      \\ \hline
	   $\ttr_{A^c}(U_{\Sigma})$ & 
       \raisebox{-0.5\height}{\includegraphics[width=5cm]{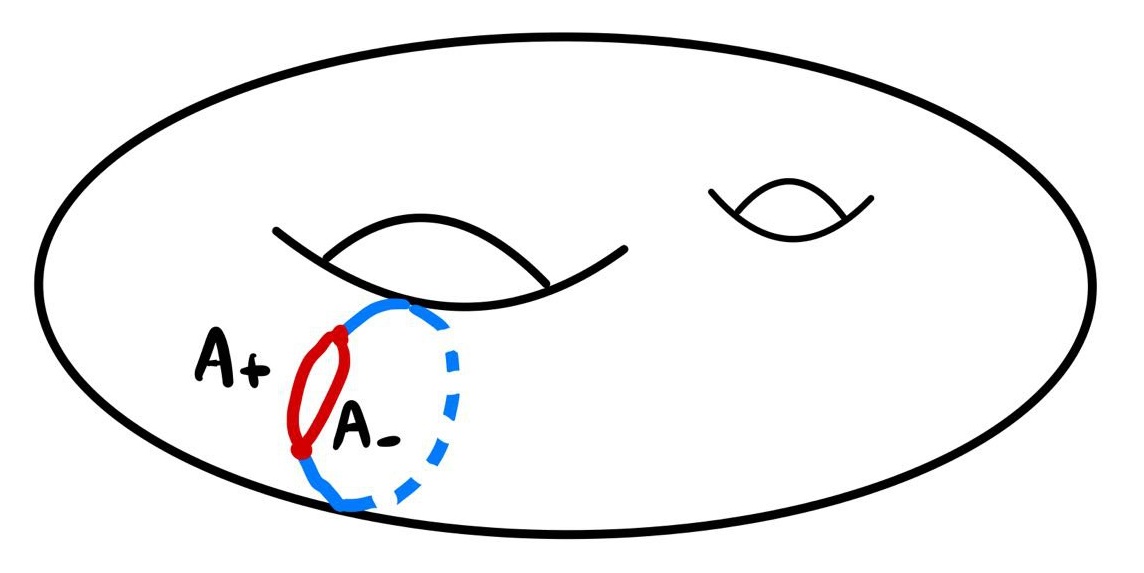}}
       \\ \hline
       $\ttr_A\big(\ttr_{A^c}(U_{\Sigma})^n\big)$ & 
       \raisebox{-0.5\height}{\includegraphics[width=5cm]{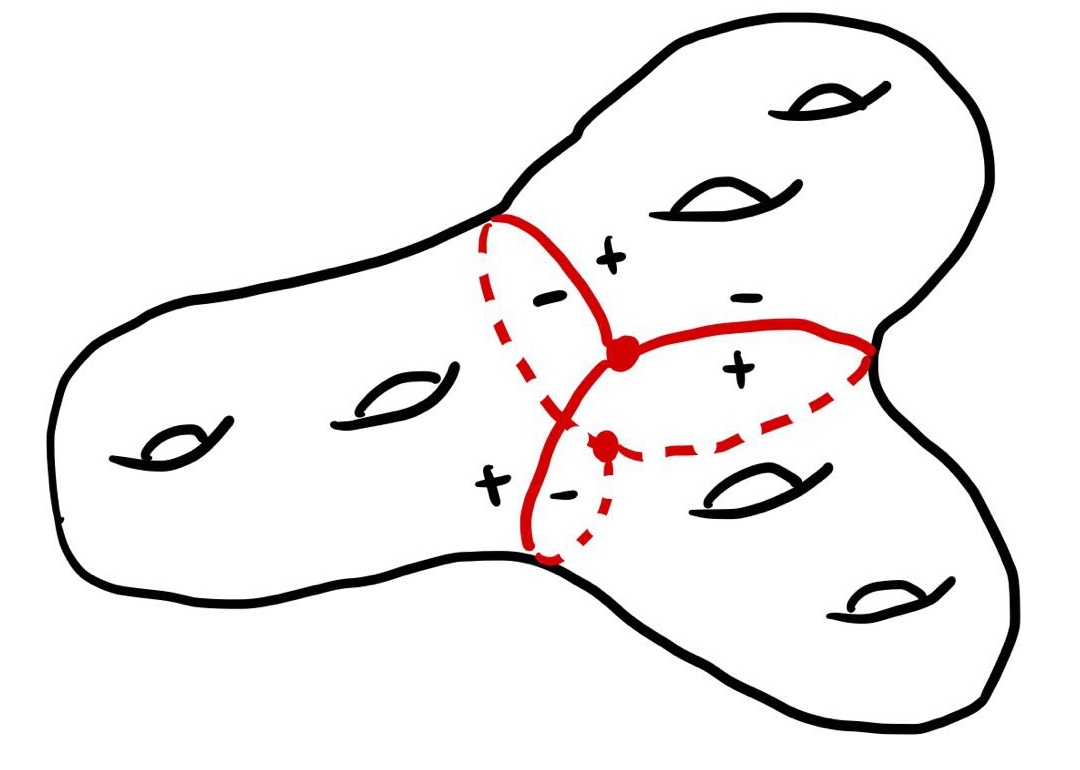}}
       \\ \hline
	\end{tabular}\end{center}

\subsection{Density Operators Represented by Imaginary-time Evolution}\label{sec-dens-op-imag-time}

We are left with the question: which density operators are also evolution operators across some space-time? To begin with, we reemphasize: our space-time metrics all have Euclidean signature and we work with \textsf{Wick rotated}/\textsf{Euclidean} QFT. More precisely, the action functional~$S_{\mm{EQFT}}(\phi)$, which usually involves itself the underlying space-time metric, is real-valued and the exponential weight~$\exp(-S_{\mm{EQFT}}(\phi))$ positive in (\ref{eqn-heu-gen-amp}) and (\ref{eqn-heu-gen-part-func}). The evolution operators defined in this way is said to be in \textsf{imaginary time}. For actual Lorentzian QFT one must use~$\exp(\ii S_{\mm{QFT}})$ instead of~$\exp(-S_{\mm{EQFT}})$, where the action~$S_{\mm{QFT}}$ is also written in terms of the Lorentzian metric.

Next, which imaginary-time evolution operators are trace class and nonnegative, serving as candidates for a density operator? In this regard it would be instructive to introduce the third axiom of Segal's:
	\begin{description}
	  \item[(adjoint)] let~$\Sigma^*$ denote the same space-time as~$\Sigma$ but with the identification of the initial and terminal boundaries \textit{reversed}, then
	    \begin{equation}
	      U_{\Sigma^*}=U_{\Sigma}^{\dagger},
	      \label{}
	    \end{equation}
	    where~$\dagger$ denotes the operator adjoint.
	\end{description}
	From this, together with the two previous axioms, one could see that Euclidean evolution across a cobordism naturally corresponds to \textit{Hilbert-Schmidt} operators. Indeed,
	\begin{equation}
	  \ttr_{\mathcal{H}}\big(U_{\Sigma}^{\dagger}U_{\Sigma}\big)=\ttr_{\mathcal{H}}\big(U_{\Sigma^*\circ \Sigma}\big)=\mathcal{Z}\big((\Sigma^*\circ \Sigma)^{\vee}\big)\quad (<\infty),
	  \label{}
	\end{equation}
	where~$(\Sigma^*\circ \Sigma)^{\vee}$ is the ``double'' of~$\Sigma$. This ``shows'' that~$U_{\Sigma}$ is Hilbert-Schmidt. Immediately from the same expression one also sees that a cobordism of the form~$\Sigma^*\circ \Sigma$ naturally gives a nonnegative trace class operator which is what we wanted. Such a cobordism has, in other words, a \textit{reflection symmetry} exchanging its initial and terminal boundaries, and whose invariant set is exactly isomorphic to a copy of~$X$ as a Riemannian submanifold.

\begin{def7}
	  A somewhat technical requirement is that the boundaries of~$\Sigma$ be geodesic so that its double would be smooth.
	\end{def7}

	\begin{def7}
	  For certain models the (ordinary) Segal axioms have been fully and rigorously constructed. See, for example, the work \cite{Lin} of the second author and also \cite{GKRV}. In these cases~$U_{\Sigma}$ is rigorously shown to be Hilbert-Schmidt. The main difficulty lies in defining (\ref{eqn-heu-gen-amp}) and (\ref{eqn-heu-gen-part-func}) rigorously, in which case it implies~$\mathcal{Z}(\check{\Sigma})<\infty$, and also in showing the generic composition axiom out of these definitions.
	\end{def7}

    Now we give some examples.

    \begin{exxx}
	  Assume that one is given a \textsf{Hamiltonian operator}~$H$ acting on~$\mathcal{H}_X$, and let~$\beta>0$. Then the so-called \textsf{thermal state} (or \textsf{Gibbs state}) at \textsf{inverse temperature}~$\beta$ is given by the density operator
	  \begin{equation}
	    \rho_{\beta}\defeq \frac{\me^{-\beta H}}{\ttr_{\mathcal{H}}(\me^{-\beta H})}.
	    \label{}
	  \end{equation}
	  The operator~$\me^{-\beta H}$ gives the evolution over imaginary time~$\ii\beta$ and could be represented by a path integral (\ref{eqn-heu-gen-amp}) over the cylinder~$\Sigma_{\beta}:=X\times [0,\beta]$ with a specific action related to~$H$. Now we denote by~$\check{\Sigma}_{n,\beta}$ the closed surface obtained by gluing~$n$ copies of~$\check{\Sigma}_{\beta}\setminus A$ cyclicly along the slit~$A$, as explained at the end of the last subsection. Then the replica trick yields
	  \begin{equation}
	    \ttr_A\big(\ttr_{A^c}(\rho_{\beta})^n\big)=\frac{\mathcal{Z}(\check{\Sigma}_{n,\beta})}{\mathcal{Z}(\check{\Sigma}_{\beta})^n}.
	    \label{}
	  \end{equation}
	  We remark again that~$\check{\Sigma}_{n,\beta}$ comes equipped with the ``replica metric'' which is induced from~$\Sigma_{\beta}$.
	\end{exxx}

    \begin{exxx}\label{exp-ground-state-disj-cob}
	  Here we consider cobordisms which are disconnected,~$\Sigma=\Sigma_1\sqcup \Sigma_2$, with~$\partial\Sigma_1\cong \partial\Sigma_2\cong X$. We mark~$\partial\Sigma_1$ as initial and~$\partial\Sigma_2$ as terminal. In this case
	  \begin{equation}
	    \conf(\Sigma)=\conf(\Sigma_1)\times \conf(\Sigma_2),\quad\quad \mathcal{L}_{\conf(\Sigma)}=\mathcal{L}_{\conf(\Sigma_1)}\otimes \mathcal{L}_{\conf(\Sigma_2)},
	    \label{}
	  \end{equation}
	  and assume (a very weak form of) \textsf{locality} of the action (no interaction between disjoint space-times), that is,
	  \begin{equation}
	    S_{\Sigma}(\phi)=S_{\Sigma_1}(\phi_1)+S_{\Sigma_2}(\phi_2),
	    \label{}
	  \end{equation}
	  with~$\phi=(\phi_1,\phi_2)$, then we have formally
	  \begin{align}
	    \mathcal{A}_{\Sigma}(\psi,\varphi)&=\iint_{ \left\{ \substack{\phi_1\in\conf(\Sigma_1), \\ \phi_1|_{\partial\Sigma_1}=\varphi} \right\}\times \left\{ \substack{\phi_2\in\conf(\Sigma_2), \\ \phi_2|_{\partial\Sigma_2}=\psi} \right\}}
	    \me^{-S_{\Sigma_1}(\phi_1)}\me^{-S_{\Sigma_2}(\phi_2)}
	    \,\dd\mathcal{L}_{\conf(\Sigma_1)}(\phi_1)\,\dd \mathcal{L}_{\conf(\Sigma_2)}(\phi_2)\nonumber \\
	    &=\mathcal{A}_{\Sigma_1}(\varphi)\mathcal{A}_{\Sigma_2}(\psi), \label{eqn-disj-union-axiom}
	  \end{align}
	  where~$\mathcal{A}_{\Sigma_1}$ and~$\mathcal{A}_{\Sigma_2}$ are the amplitudes associated to~$\Sigma_1$ and~$\Sigma_2$ individually, defined by the same formula (\ref{eqn-heu-gen-amp}) in the case of having just one boundary component (sometimes (\ref{eqn-disj-union-axiom}) itself is taken as an axiom, see \cite{Gaw} page 766). We see that in this case~$U_{\Sigma}$ is a rank-one operator. Written more suggestively,
	\begin{equation}
	  U_{\Sigma}=\big|\mathcal{A}_{\Sigma_2}\big\rangle \big\langle \mathcal{A}_{\Sigma_1}\big|.
	  \label{}
	\end{equation}
    	We normalize by putting~$\rho_{\Sigma}:=U_{\Sigma}/\ttr(U_{\Sigma})$, then following the replica trick we have
	\begin{equation}
	  \ttr_A\big(\ttr_{A^c}(\rho_{\Sigma})^n\big)=\frac{\mathcal{Z}(\check{\Sigma}_n)}{\mathcal{Z}(\check{\Sigma})^n},
	  \label{}
	\end{equation}
	where~$\check{\Sigma}$ is the closed surface obtained by gluing~$\Sigma_1$ with~$\Sigma_2$ along~$X$, and~$\check{\Sigma}_n$ by gluing~$n$ copies of~$\check{\Sigma}\setminus A$ along the slit~$A$, equipped with the induced metric. In particular, if~$\Sigma_1=\Sigma_2=\mb{D}$, the unit disk, equipped with a flat-at-the-boundary metric (namely it could be written~$|z|^{-2}|\dd z|^2$ in some complex coordinate on an annulus around~$\partial\mb{D}$), then~$\rho_{\Sigma}$ defined as above is metric independent (provided flat at the boundary) and represents (projection onto) the \textsf{vacuum state}, usually written~$|0\rangle\langle 0|$ (see \cite{Gaw} page 768).
	\end{exxx}

\newpage
\section{Geometric Lemmas}\label{sec-geo-lemm}
	To begin with, we record some basic conformal relations that will be used repeatedly throughout the paper. These relations are local in nature as they concern only what happens in an arbitrarily small neighborhood around each point. Namely, let~$\Sigma$ be a smooth \textit{surface} and denote by~$\dd V_g$,~$\nabla_g$,~$K_g$,~$\Delta_g$ and $\nabla_v^g$ respectively the volume (area) form, gradient, Gauss curvature (half the scalar curvature), Laplacian and the Levi-Civita covariant derivative in the direction~$v\in T_z\Sigma$ under the Riemannian metric~$g$, and let~$h\in C^{\infty}(\Sigma)$, then we have
\begin{align}
      \dd V_{\me^{2h}g}&=\me^{2h}\dd V_g, \label{eqn-scale-vol}\\
      \nabla_{\me^{2h}g}&=\me^{-2h}\nabla_g, \label{eqn-scale-nabla}\\
      K_{\me^{2h}g}&=\me^{-2h}\left( -\Delta_g h+K_g \right),
      \label{eqn-liouville-eqn}\\
      \Delta_{\me^{2h}g}&=\me^{-2h}\Delta_g. \label{eqn-scale-lap} \\
      \nabla_v^{\me^{2h}g} X&=\nabla_v^g X+ (vh)X+(Xh)v-\ank{v,X}_g \nabla_g h. \label{eqn-scale-cov-deri}
    \end{align}
The relation (\ref{eqn-liouville-eqn}) is also called \textsf{Liouville's equation}.

\subsection{Technical Remarks on Conical Singularities}\label{sec-conic-remark}

The principal local regularity result around a conical singularity in our sense (definition \ref{def-singularities}) is the following obtained by Troyanov \cite{Troy2}.

  \begin{lemm}
    [\cite{Troy2} proposition 3.2] \label{lemm-reg-metric-pot} Let~$\Sigma$ be a closed Riemann surface and~$\tilde{g}$ a generalized conformal metric on~$\Sigma$ with a conical singularity at~$z_0\in \Sigma$ of exponent~$\gamma$. Then there exists a complex coordinate~$w$ defined on a neighborhood~$U\ni z_0$, with~$w(z_0)=0$, such that
    \begin{equation}
      \tilde{g}|_{U\setminus z_0}=|w|^{2\gamma}\me^{2 \varphi_{U,w}}|\dd w|^2,
      \label{}
    \end{equation}
    such that~$\varphi_{U,w}$ satisfies the conditions of definition \ref{def-singularities} as well as
    \begin{equation}
      \varphi_{U,w}(z)=\varphi_{U,w}(z_0)+\mathcal{O}(|w|^{2\gamma+2}),\quad\textrm{and}\quad \partial_z \varphi_{U,w},~\partial_{\ol{z}} \varphi_{U,w}=\mathcal{O}(|w|^{2\gamma+1}).
      \label{eqn-scale-metric-pot}
    \end{equation}
    In particular, $\partial_r \varphi_{U,w}=\mathcal{O}(|w|^{2\gamma+1})$ where $r$ is the radial coordinate~$r(z)=|w(z)|$ under $|\dd w|^2$. \hfill~$\Box$
  \end{lemm}

   Troyanov showed further that a neighborhood (still denoted~$U$) of the conical singularity could in fact be equivalently and more intrinsically described by a set of \textsf{polar coordinates}. That is, there is a map
   \begin{equation}
	  h:[0,\varrho]\times \mb{S}_{\Theta}^1 \lto U,
	  \label{}
	\end{equation}
	where~$\mb{S}_{\Theta}^1$ denotes the Riemannian circle with perimeter~$\Theta=2\pi(\gamma+1)$, such that
	\begin{enumerate}[(a)]
	  \item $h(r,\theta)=z_0$ iff~$r=0$;
	\item $h|_{(0,\varrho]\times \mb{S}^1_{\Theta}}$ is a locally bi-Lipschitz homeomorphism onto~$U\setminus z_0$; 
	\item we have
	  \begin{equation}
	    h^*(\tilde{g}|_{U\setminus z_0})=\dd r^2+\omega(r,\theta)^2\dd \theta^2,
	    \label{}
	  \end{equation}
	where~$\omega:(0,\varrho]\times \mb{S}_{\Theta}^1\lto \mb{R}$ is a function such that~$0<c_1\le \omega(r,\theta)\le c_2$ for some constants~$c_1$,~$c_2$ for all~$(r,\theta)$ and~$\lim_{r\to 0}\omega(r,\theta)/r =1$ for all~$\theta\in \mb{S}_{\Theta}^1$.
	\end{enumerate}
By the last requirement we have
\begin{equation}
  \Theta(r)\defeq \int_{\mb{S}^1}^{}\omega(r,\theta)\,\dd\theta\asymp r\Theta=2\pi r(\gamma+1),\quad\quad r\to 0,
  \label{}
\end{equation}
and thus~$2\pi(\gamma+1)$ is sometimes called the \textsf{cone angle}. Moreover, the regularity of the function~$\omega$ as well as~$h$ itself could be deduced from the regularity of the curvature of~$\tilde{g}$ on~$U\setminus z_0$ (\cite{Troy2} theorem 4.1). There are certain recent works in the literature which begin naturally with this latter point of view (\cite{AKR} for example). But we stick to definition \ref{def-singularities} throughout this article although utilization of polar coordinates may (or may not) aide with certain proofs.

\begin{def7}\label{rem-reg-met-pot-cone}
  In the literature people have considered more restrictive classes of allowable conical metrics by assuming more regularity on the regular metric potential (see \cite{Kalvin} section 2.1). The principle is to respect the correct scaling property under dilations centered at the cone points. Especially for the conical metrics with \textit{constant curvature}, the corresponding regular metric potentials (against coordinate metrics) are shown to be \textsf{dilation analytic}. This translates, in our notations, roughly into saying that~sufficiently near~$z_0$, $\varphi_{U,w}$ would be a real analytic function of~$|w|^{\gamma+1}$, extendable over a neighborhood of zero, in each direction respectively (this is clearly an enhancement of (\ref{eqn-scale-metric-pot})). However, due to the simplicity of our method we only need a very rough scaling property that is already deducible from lemma \ref{lemm-reg-metric-pot}, which we treat in the section below.
\end{def7}

\subsection{Dilation Properties at the Cone Point}\label{sec-geo-scaling}

     \begin{lemm}\label{lemm-scale-dist-sing}
Let~$\Sigma$ be a closed Riemann surface and~$\tilde{g}$ a generalized conformal metric on~$\Sigma$ with a conical singularity at~$z_0\in \Sigma$ of exponent~$\gamma$. Then 
	  \begin{equation}
	    \tilde{r}(z)=\frac{\me^{\varphi(z_0)}}{\gamma+1}r(z)^{\gamma+1}+\mathcal{O}\left(r^{3\gamma+3}\right),\quad\quad  z\to z_0,
	    \label{}
	  \end{equation}
	  where ~$\tilde{r}(z):=d_{\tilde{g}}(z,z_0)$,~$r(z):=|w(z)|$,~$z\ne z_0$, $\varphi:=\varphi_{U,w}$ as in lemma \ref{lemm-reg-metric-pot} and $w$ is the complex coordinate from lemma \ref{lemm-reg-metric-pot}. In particular, we see that~$\tilde{r}(z)\asymp r(z)^{\gamma+1}$ as~$z\to z_0$.
	\end{lemm}

	\begin{proof}
	  By (\ref{eqn-scale-metric-pot}) we have
	  \begin{equation}
	    0<\me^{\varphi(z_0)}-\alpha r(z)^{2\gamma+2}\le \me^{\varphi(z)}\le \me^{\varphi(z_0)}+\alpha r(z)^{2\gamma+2}
	    \label{eqn-conf-factor-sandwich}
	  \end{equation}
	  for some~$\alpha>0$, in some smaller neighborhood~$U_1\subset U$. Now fix~$z\in U_1$, let~$c(s)=w^{-1}(s w(z)/|w(z)|)$ be defined for~$s\in [0,r(z)]$ (the unit speed geodesic under~$|\dd w|^2$). Then~$\tilde{r}(z)$ is majorized by the length of~$c$ under~$\tilde{g}$, namely
	  \begin{equation}
	    \tilde{r}(z)\le \int_{0}^{r}|c'(s)|_{\tilde{g}}\,\dd s=\int_{0}^{r} s^{\gamma}\me^{\varphi(c(s))}\,\dd s\le \frac{\me^{\varphi(z_0)}r^{\gamma+1}}{\gamma+1}+\alpha \frac{ r^{3\gamma+3}}{3\gamma+3},
	    \label{eqn-con-dist-finite}
	  \end{equation}
	  by (\ref{eqn-conf-factor-sandwich}). On the other hand, the length under~$\tilde{g}$ of any curve inside~$U_1$ is bounded below by its length under the metric
	  \begin{equation}
	    g_1\defeq r^{2\gamma}\big( \me^{\varphi(z_0)}-\alpha r^{2\gamma+2} \big)^2\cdot |\dd w|^2 =f(r)^2\cdot |\dd w|^2.
	    \label{}
	  \end{equation}
	  defined on~$U_1$. By (\ref{eqn-scale-cov-deri}) we know that~$c$ would now reparametrize into a geodesic under~$g_1$ (minimizing since $g_1$ is radial with respect to the geodesic coordinates of $g$). Therefore we obtain
	  \begin{equation}
	    \tilde{r}(z)\ge \int_{0}^{r} |c'(s)|_{g_1}\,\dd s=\frac{\me^{\varphi(z_0)}r^{\gamma+1}}{\gamma+1}-\frac{\alpha r^{3\gamma+3}}{3\gamma+3}.
	    \label{}
	  \end{equation}
	  This gives the result.
	\end{proof}

    \begin{def7}
  In particular, (\ref{eqn-con-dist-finite}) also shows~$\tilde{r}(z)<\infty$, which is not necessarily true \textit{a priori}.
\end{def7}

 From this lemma it follows immediately the next two corollaries.

 	\begin{corr}
	  With the same set-up and notation as lemma \ref{lemm-scale-dist-sing}, we have
	  \begin{equation}
	    r(z)=(\gamma+1)^{\frac{1}{\gamma+1}}\me^{-\frac{\varphi(z_0)}{\gamma+1}}\tilde{r}(z)^{\frac{1}{\gamma+1}} +\mathcal{O}(\tilde{r}^{2+\frac{1}{\gamma+1}}),
	    \label{}
	  \end{equation}
	  as~$z\to z_0$.
	\end{corr}

    \begin{proof}
        This follows from a simple order analysis using Newton's binomial formula.
    \end{proof}

	\begin{corr}\label{cor-ratio-radius}
	  Pick~$h\in C^{\infty}(\Sigma)$ and denote by~$\tilde{r}_h(z)$ the distance from~$z$ to~$z_0$ under the scaled metric
	  \begin{equation}
	    \tilde{g}_h|_U\defeq d_g(\bullet,z_0)^{2\gamma}\me^{2(\varphi+h)}\cdot g|_U.
	    \label{}
	  \end{equation}
	  Denote by~$\delta(\varepsilon)$ and~$\delta_h(\varepsilon)$ metric radii under~$g$ of the~$\varepsilon$-metric disks under~$\tilde{g}$ and~$\tilde{g}_h$ respectively centered at~$z_0$. Then we have
	  \begin{equation}
	    \frac{\delta(\varepsilon)}{\delta_h(\varepsilon)}\lto \exp\left(\frac{h(z_0)}{\gamma+1}\right),
	    \label{}
	  \end{equation}
	  as~$\varepsilon\to 0^+$. \hfill~$\Box$
	\end{corr}

\subsection{Log-divergent Integrals}\label{sec-int-by-part}
	In this section we collect some basic computations of ``toy integrals'' involving functions with log-divergent singularities, which will nevertheless play a fundamental role in the proofs of the main results of this paper.

	Throughout this section, we assume~$(\Sigma,g)$ is a Riemannian surface with smooth metric~$g$, let~$z_0\in \Sigma$ and let~$\sigma\in C^{\infty}(\Sigma\setminus\{z_0\})$ such that in a neighbourhood $U$ of $z_0$,
	\begin{equation}
	  \sigma(z)=\gamma \log d_g(z,z_0) +\varphi(z)
	  \label{eqn-local-form-conf-factor}
	\end{equation}
	for some~$\gamma>-1$, where $\varphi$ is a function satisfying the requirements for the metric potential in definition \ref{def-singularities} as well as (\ref{eqn-scale-metric-pot}) with~$|\dd w|^2$ replaced by~$g$. Formula (\ref{eqn-limit-lap-smooth-against-con}), however, needs much weaker assumptions. We start by recalling the following basic fact.
	\begin{lemm}\label{lemm-bounded-lap-log}
	  In the set-up as above, the function~$\Delta_g\log d_g(\bullet,z_0)$ is smooth and uniformly bounded near but not coincident with $z_0$.
	\end{lemm}

    \begin{proof}
        See appendix \ref{sec-app-poin-lelong}.
    \end{proof}

 \begin{def7}
	  A related classical result says that the perimeter~$\ell_g(\partial B_r(z_0))$ of a geodesic circle centered at~$z_0\in \Sigma$ of radius~$r$ has an asymptotics
	  \begin{equation}
	    \ell_g(\partial B_r(z_0))=2\pi r-\frac{\pi}{3}r^3 K_g(z_0)+o(r^3),\quad r\to 0^+.
	    \label{eqn-asymp-peri-circle}
	  \end{equation}
	  Here as above,~$K_g(z_0)$ is the Gauss curvature of~$g$ at~$z_0$. See, for example, \cite{dC2} page 296.
	\end{def7}

 \begin{lemm}\label{lemm-basic-renorm-log}
	  In the set-up as above, let~$\delta>0$. Then we have
	  \begin{equation}
	    \int_{\Sigma\setminus B_{\delta}(z_0,g)}^{} |\nabla_g \sigma|_g^2\,\dd V_g=-2\pi\gamma^2\log(\delta)-2\pi\gamma \varphi(z_0)
        -\int_{\Sigma\setminus B_{\delta}}^{} \sigma\Delta_g \sigma\,\dd V_g
        +\mathcal{O}(\delta^{2\min\{\gamma,0\}+2}\log \delta)
	    \label{}
	  \end{equation}
	  as~$\delta\to 0^+$, where~$B_{\delta}(z_0,g)$ denotes the geodesic disk around~$z_0$ of radius~$\delta$ under the metric~$g$.
	\end{lemm}

	\begin{proof}
	  Denote~$r(z):=d_g(z,z_0)$. Integrating by parts (Green-Stokes formula) we have
	  \begin{align*}
	    \textrm{LHS}&=\int_{\partial B_{\delta}}^{} \sigma(-\partial_r \sigma)\,\dd\ell_{g} -\int_{\Sigma\setminus B_{\delta}}^{} \sigma\Delta_g \sigma\,\dd V_g \\
	    &=-\gamma^2\underbrace{\int_{\partial B_{\delta}}^{} \log r(\partial_r \log r)\,\dd\ell_{g}}_{A}
	    -\gamma\underbrace{\int_{\partial B_{\delta}}^{} \varphi(\partial_r \log r)\,\dd\ell_{g}}_{B}
	    -\gamma\underbrace{\int_{\partial B_{\delta}}^{} \log r(\partial_r \varphi)\,\dd\ell_{g}}_{C} \\
	    &\quad\quad -\underbrace{\int_{\partial B_{\delta}}^{} \varphi(\partial_r \varphi)\,\dd\ell_{g}}_{D}-\int_{\Sigma\setminus B_{\delta}}^{} \sigma\Delta_g \sigma\,\dd V_g.
	  \end{align*}
	  By the assumptions (\ref{eqn-scale-metric-pot}) on $\varphi$ and by (\ref{eqn-asymp-peri-circle}), we have
	  \begin{align}
	    A&=\log\delta\cdot\frac{1}{\delta}(2\pi \delta+\mathcal{O}(\delta^3))=2\pi\log\delta+\mathcal{O}(\delta^2\log \delta),\\
	    B&=\frac{1}{\delta}\int_{\partial B_{\delta}}^{} \varphi\,\dd\ell_{g} = 2\pi \varphi(z_0)+\mathcal{O}(\delta^2),\\
	    C&=\log\delta\int_{\partial B_{\delta}}^{} \partial_r\varphi\,\dd\ell_{g}=\mathcal{O}(\delta^{2\gamma+2}\log \delta),\\
	    D&=\mathcal{O}(\delta^{2\gamma+2}).
	    \label{}
	  \end{align}
	  Adding them all up, we obtain the result. Note $2\gamma+2$ may be small but is positive, so $\mathcal{O}(\delta^{2\gamma+2}\log \delta)=o(1)$.
	\end{proof}

  \begin{corr}\label{cor-quad-blow-tilde}
  In the set-up as above, let~$\tilde{g}$ be a generalised conformal metric representing $D=\gamma z_0$ with $\gamma>-1$, written in the form (\ref{eqn-conic-sing-local-form}). Suppose~$g$ is a smooth conformal metric on~$\Sigma$ and~$\tilde{g}=\me^{2\sigma}g$ for some~$\sigma\in C^{\infty}(\Sigma\setminus \{z_0\})$, so that $\sigma$ is of the form (\ref{eqn-local-form-conf-factor}). Then
  \begin{align}
    \int_{\Sigma\setminus B_{\varepsilon}(z_0,\tilde{g})}^{} |\nabla_g \sigma|^2\,\dd V_g&= -\frac{2\pi \gamma^2}{\gamma+1}\big[\log(\varepsilon)+\log(\gamma+1)-\varphi(z_0)\big]
    -2\pi\gamma \varphi(z_0)
        -\int_{\Sigma\setminus z_0}^{} \sigma\Delta_g \sigma\,\dd V_g
    +o(1)
    \label{}
  \end{align}
  as~$\varepsilon\to 0$, where~$B_{\varepsilon}(z_0,\tilde{g})$ denotes the geodesic ball around~$z_0$ of radius~$\varepsilon$ under the metric~$\tilde{g}$.
\end{corr}

	\begin{proof}
	  This is because~$B_{\varepsilon}(z_0,\tilde{g})$ has radius~$\delta(\varepsilon)\sim((\gamma+1)\varepsilon \me^{-\varphi(z_0)})^{1/(\gamma+1)}$ under the metric~$g$ by lemma \ref{lemm-scale-dist-sing}.
	\end{proof}

\begin{lemm}\label{lemm-linear-sing-green-sto}
	  In the set-up as lemma \ref{lemm-basic-renorm-log}, pick further~$h\in C^{\infty}(\Sigma)$. Then we have
	  \begin{align}
	    \lim_{\delta\to 0^+}\int_{\Sigma\setminus  B_{\delta}(z_0,g)}^{} \big(
    \ank{\nabla_g h,\nabla_g\sigma}_g +(\Delta_g h)\sigma\big)\,\dd V_g &=0, \label{eqn-limit-lap-smooth-against-con}\\
\lim_{\delta\to 0^+}\int_{\Sigma\setminus  B_{\delta}(z_0,g)}^{} \big(
    \ank{\nabla_g h,\nabla_g\sigma}_g +h(\Delta_g\sigma)\big)\,\dd V_g &=-2\pi\gamma h(z_0).
	    \label{}
	  \end{align}
   Here for (\ref{eqn-limit-lap-smooth-against-con}) we only need the weaker assumption that~$\varphi$ be bounded over~$U$ in the expression (\ref{eqn-local-form-conf-factor}) for~$\sigma$. Moreover, the 2-form~$(
    \ank{\nabla_g h,\nabla_g\sigma}_g +h(\Delta_g\sigma))\dd V_g$ is conformally invariant, that is,
    \begin{equation}
      \big(
      \sank{\nabla_{\me^{2\varphi}g} h,\nabla_{\me^{2\varphi}g}\sigma}_{\me^{2\varphi}g} +h(\Delta_{\me^{2\varphi}g}\sigma)\big)\dd V_{\me^{2\varphi}g} =\big(
    \ank{\nabla_g h,\nabla_g\sigma}_g +h(\Delta_g\sigma)\big)\dd V_g
      \label{eqn-conf-inv-2-form}
    \end{equation}
    for any~$\varphi\in C^{\infty}(\Sigma\setminus \{z_0\})$. Therefore we have in particular
    \begin{equation}
      \lim_{\varepsilon\to 0^+}\int_{\Sigma\setminus  B_{\delta}(z_0,g)}^{} \big(
      \sank{\nabla_{\tilde{g}} h,\nabla_{\tilde{g}}\sigma}_{\tilde{g}} +h(\Delta_{\tilde{g}}\sigma)\big)\,\dd V_{\tilde{g}} =-2\pi\gamma h(z_0),
      \label{}
    \end{equation}
    where~$\tilde{g}$ is as defined in corollary \ref{cor-quad-blow-tilde}.
	\end{lemm}

	\begin{proof}
	  We apply Green-Stokes. The first integral boils down to
	  \begin{align*}
	    \int_{\partial B_{\delta}}^{}(-\partial_r h)\sigma\,\dd\ell_g &=\gamma\int_{\partial B_{\delta}}^{}(-\partial_r h)\log r\,\dd\ell_g+\int_{\partial B_{\delta}}^{}(-\partial_r h)\varphi\,\dd\ell_g= \mathcal{O}(\delta\log \delta)+\mathcal{O}(\delta),
	  \end{align*}
	  and the second to
	  \begin{align*}
	    \int_{\partial B_{\delta}}^{}h(-\partial_r \sigma)\,\dd\ell_g &=-\gamma\int_{\partial B_{\delta}}^{}h\cdot\frac{1}{r}\,\dd\ell_g+\int_{\partial B_{\delta}}^{}h(-\partial_r \varphi)\,\dd\ell_g\\
	    &=-2\pi\gamma h(z_0)+\mathcal{O}(\delta^{2\min\{\gamma,0\}+2}),
	  \end{align*}
	  as~$\delta\to 0^+$, by (\ref{eqn-scale-metric-pot}) and (\ref{eqn-asymp-peri-circle}). Equality (\ref{eqn-conf-inv-2-form}) follows directly from (\ref{eqn-scale-vol}), (\ref{eqn-scale-nabla}) and (\ref{eqn-scale-lap}) and we obtain the rest of the lemma.
	\end{proof}

The following lemma is not used in the main proof, but is important regarding remark \ref{rem-extra-quad-term}.

\begin{lemm}\label{lemm-quadratic-sing}
    Under the assumptions at the beginning of this section, let~$0<\delta_1(\varepsilon)<\delta_2(\varepsilon)$ be two positive functions of~$\varepsilon$ such that~$\delta_{i}(\varepsilon)\to 0$ as~$\varepsilon\to 0^+$,~$i=1$,~$2$, and~$\delta_2/\delta_1\to Q>0$.  Then
    \begin{equation}
      \lim_{\varepsilon\to 0^+}\int_{B_{\delta_2(\varepsilon)}(z_0,g)\setminus B_{\delta_1(\varepsilon)}(z_0,g)}^{}|\nabla_g \sigma|_g^2\,\dd V_g=2\pi \gamma^2\log Q.
      \label{}
    \end{equation}
  \end{lemm}

\begin{proof}
    We abbreviate the disks~$B_{\delta_i(\varepsilon)}(z_0,g)$ as~$B(\delta_i)$,~$i=1$,~$2$. Again integrating by parts we have,
    \begin{equation}
      \int_{B(\delta_2)\setminus B(\delta_1)}^{}|\nabla_g \sigma|_g^2\,\dd V_g =-\int_{B(\delta_2)\setminus B(\delta_1)}^{} \sigma \Delta_g \sigma\,\dd V_g
      +\int_{\partial B(\delta_2)}^{}\sigma(\partial_r \sigma)\,\dd\ell_g
      -\int_{\partial B(\delta_1)}^{}\sigma(\partial_r \sigma)\,\dd\ell_g.
      \label{}
    \end{equation}
    Since~$\sigma$ is locally integrable and~$\Delta_g \sigma$ is bounded on~$\Sigma\setminus\{z_0\}$, we have
    \begin{equation}
      \int_{B(\delta_2)\setminus B(\delta_1)}^{} \sigma \Delta_g \sigma\,\dd V_g \quad\xlongrightarrow{\varepsilon\to 0^+}\quad 0.
      \label{}
    \end{equation}
    Next by (\ref{eqn-local-form-conf-factor}) and (\ref{eqn-scale-metric-pot}),
    \begin{align*}
      &\int_{\partial B(\delta_2)}^{}\varphi(\partial_r \varphi)\,\dd\ell_g,~
      \int_{\partial B(\delta_1)}^{}\varphi(\partial_r \varphi)\,\dd\ell_g= \mathcal{O}(\delta^{2\gamma+2})&&\xlongrightarrow{\varepsilon\to 0^+}\quad 0,\\
      &\int_{\partial B(\delta_2)}^{}\varphi(\partial_r \log r)\,\dd\ell_g,~
      \int_{\partial B(\delta_1)}^{}\varphi(\partial_r \log r)\,\dd\ell_g&&\xlongrightarrow{\varepsilon\to 0^+}\quad 2\pi \varphi(z_0).\\
      &\int_{\partial B(\delta_2)}^{}\log r(\partial_r \varphi)\,\dd\ell_g,~
      \int_{\partial B(\delta_1)}^{}\log r(\partial_r \varphi)\,\dd\ell_g
      = \mathcal{O}(\delta^{2\gamma+2}\log \delta)&&\xlongrightarrow{\varepsilon\to 0^+}\quad 0.
    \end{align*}
    Therefore the only thing left is
    \begin{equation}
     \int_{\partial B(\delta_2)}^{}\frac{\log r}{r}\,\dd\ell_g -\int_{\partial B(\delta_1)}^{} \frac{\log r}{r}\,\dd\ell_g   \sim 2\pi \log\Big(\frac{\delta_2}{\delta_1}\Big)
      \quad\xlongrightarrow{\varepsilon\to 0^+}\quad 2\pi \log Q.
      \label{}
    \end{equation}
    Adding up all the above, we obtain the result. 
  \end{proof}

\section{Renormalization Procedure}\label{sec-renorm-tech}

\subsection{Consistency}

As we could see from definitions \ref{def-main-renorm-anomaly} and \ref{def-renom-part-func} that a particular reference smooth metric $g$ was chosen to make the definition. Now we show that a different conformal reference metric would in fact give the same partition function for the target metric $\tilde{g}$ and hence $\mathcal{Z}(\Sigma,\tilde{g})$ is invariantly defined.

\begin{lemm}\label{lemm-consist}
	  Consider a conformal field theory with central charge~$c\in\mb{R}$ on the Riemann surface~$\Sigma$ and let~$\tilde{g}$ be a generalized conformal metric representing~$D=\sum_{j=1}^p \gamma_j z_j$ with~$\gamma_j>-1$. Also let~$g_1$ and~$g_0$ be two smooth conformal metrics such that~$g_1=\me^{2h}g_0$ for some~$h\in C^{\infty}(\Sigma)$, and~$\tilde{g}=\me^{2\sigma}g_1$ for~$\sigma\in C^{\infty}(\Sigma\setminus \supp D)$. Then if one of $\mathcal{R}A_{\Sigma}(\tilde{g},g_0)$ and $\mathcal{R}A_{\Sigma}(\tilde{g},g_1)$ defined by (\ref{eqn-def-renorm-polya}) is finite, so is the other, and we have
	  \begin{equation}
	    \mathcal{R}A_{\Sigma}(\tilde{g},g_0)-\mathcal{R}A_{\Sigma}(\tilde{g},g_1)=   A_{\Sigma}(g_1,g_0),
	    \label{}
	  \end{equation}
   	 where $A_{\Sigma}(g_1,g_0)$ is defined by (\ref{eqn-def-ord-poly-anomaly}). Hence definition (\ref{eqn-def-renom-part-func}) is independent of the reference metric chosen.
	\end{lemm}

 \begin{proof}
	  Without loss of generality suppose~$D=\gamma z_0$ with~$\gamma>-1$. Denote, for simplicity, by~$\nabla_i$,~$|\cdot|_i$,~$K_i$ and~$\dd V_i$ respectively the gradient, metric norm, Gauss curvature and area form under the metric~$g_i$,~$i=0$,~$1$. Then by definition and the relations (\ref{eqn-scale-vol}) --- (\ref{eqn-liouville-eqn}), valid on~$\Sigma\setminus  B_{\varepsilon}(z_0,\tilde{g})$ for any~$\varepsilon>0$, we have
	  \begin{align}
	    \mathcal{R}A_{\Sigma}(\tilde{g},g_1)&=\frac{1}{24\pi}\lim_{\varepsilon\to 0^+}\Big[ 
	      \int_{\Sigma\setminus  B_{\varepsilon}(z_0,\tilde{g})}^{}(|\nabla_{1} \sigma|_{1}^2+2K_{1}\sigma)\dd V_{1} +\frac{2\pi\gamma^2}{\gamma+1}\log(\varepsilon)
    \Big]\\
    &=\frac{1}{24\pi}\lim_{\varepsilon\to 0^+}\Big[ 
      \int_{\Sigma\setminus  B_{\varepsilon}(z_0,\tilde{g})}^{}( |\nabla_0 \sigma|_0^2 -2(\Delta_0 h)\sigma +2K_0 \cdot\sigma
      )\dd V_{0} +\frac{2\pi\gamma^2}{\gamma+1}\log(\varepsilon)
    \Big],\\
    \mathcal{R}A_{\Sigma}(\tilde{g},g_0)&=\frac{1}{24\pi}\lim_{\varepsilon\to 0^+}\Big[ 
      \int_{\Sigma\setminus  B_{\varepsilon}(z_0,\tilde{g})}^{}\big(|\nabla_{0}(h+ \sigma)|_{0}^2+2K_{0}(h+\sigma)\big)\dd V_{0} +\frac{2\pi\gamma^2}{\gamma+1}\log(\varepsilon)
    \Big].
  \end{align}
  Therefore
  \begin{align*}
    \mathcal{R}A_{\Sigma}(\tilde{g},g_0)-\mathcal{R}A_{\Sigma}(\tilde{g},g_1)&=\frac{1}{24\pi}\int_{\Sigma}^{}(|\nabla_{0}h|_{0}^2+2K_{0}\cdot h)\,\dd V_{0}\\
    &\quad\quad +\frac{1}{12\pi}\underbrace{\lim_{\varepsilon\to 0^+}\int_{\Sigma\setminus  B_{\varepsilon}(z_0,\tilde{g})}^{} \big(
    \ank{\nabla_0 h,\nabla_0\sigma}_0 +(\Delta_0 h)\sigma\big)\,\dd V_0}_{=~0}\\
    &=\frac{1}{24\pi}\int_{\Sigma}^{}(|\nabla_{0}h|_{0}^2+2K_{0}\cdot h)\,\dd V_{0}
  \end{align*}
  by (\ref{eqn-limit-lap-smooth-against-con}) since $\sigma$ has the form (\ref{eqn-local-form-conf-factor}) with $\varphi$ bounded with respect to $g_0$. 
	\end{proof}

\subsection{Regularized Curvature and Anomaly}
In this subsection we note another kind of regularized anomaly which concerns smooth (bounded) scaling of a singular metric, rather than singular scaling of a smooth metric which was in some sense what we did above. This quantity will also play a role in the main result proposition \ref{prop-main-conical-scaling}.

	\begin{lemm}\label{lemm-smooth-scale-con}
	  Let~$\Sigma$ be a closed Riemann surface and~$\tilde{g}$ a generalized conformal metric representing~$D=\sum_{j=1}^p \gamma_j z_j$ with~$\gamma_j>-1$. Now suppose~$h\in C^{\infty}(\Sigma)$ and consider the scaled metric~$\me^{2h}\tilde{g}$ on~$\Sigma\setminus\supp D$. Then
  \begin{equation}
    \mathcal{R}A_{\Sigma}(\me^{2h}\tilde{g},\tilde{g})\defeq \frac{1}{24\pi}
      \int_{\Sigma\setminus \supp D}\big(|\nabla_{\tilde{g}} h|_{\tilde{g}}^2+2K_{\tilde{g}}h\big)\,\dd V_{\tilde{g}} <\infty.
    \label{eqn-def-renorm-polya-conic}
  \end{equation}
	\end{lemm}

	\begin{proof}
	  Choose a smooth background conformal metric~$g$ and write~$\tilde{g}=\me^{2\sigma}g$,~$\sigma\in C^{\infty}(\Sigma\setminus\supp D)$. As in lemma \ref{lemm-linear-sing-green-sto}, the 2-form~$|\nabla_{\bullet}h|_{\bullet}^2\,\dd V_{\bullet}$ is conformally invariant, and the Gauss curvature transforms as
	  \begin{equation}
	    K_{\tilde{g}}=\me^{-2\sigma}(-\Delta_g \sigma+K_g),\quad \textrm{on }\Sigma\setminus\supp D.
	    \label{}
	  \end{equation}
	  Therefore
	  \begin{align}
	    \int_{\Sigma\setminus \supp D}(|\nabla_{\tilde{g}} h|_{\tilde{g}}^2+2K_{\tilde{g}}h)\,\dd V_{\tilde{g}} =
	    \int_{\Sigma\setminus \supp D}(|\nabla_{g}h|_{g}^2+2(-\Delta_g \sigma+K_g)h)\,\dd V_{g}<\infty, \label{eqn-conf-change-to-smooth}
	  \end{align}
	  because of lemma \ref{lemm-bounded-lap-log}.
	\end{proof}

\begin{def7}
	  As a by-product, we also see that the RHS of (\ref{eqn-conf-change-to-smooth}), which is expressed in terms of the background smooth metric~$g$, is independent of~$g$.
	\end{def7}

\begin{deef}\label{def-renorm-polya-conic}
	  In the situation as lemma \ref{lemm-smooth-scale-con}, we call (\ref{eqn-def-renorm-polya-conic}) the \textsf{regular Polyakov anomaly} of~$\me^{2h}\tilde{g}$ against~$\tilde{g}$.
	\end{deef}

\begin{def7}
  The distribution (or current, more precisely) on~$\Sigma$ which we denote by~$K_{\tilde{g}}\dd V_{\tilde{g}}$, defined by setting
  \begin{equation}
    \ank{K_{\tilde{g}}\dd V_{\tilde{g}}, \psi}\defeq \int_{\Sigma\setminus \supp D}^{}\psi K_{\tilde{g}}\,\dd V_{\tilde{g}}
    \label{}
  \end{equation}
  for all~$\psi\in C^{\infty}(\Sigma)$, is usually called the \textsf{regularized Gauss curvature}. This is in $L^1(\Sigma, g)\,\dd V_g$ for any smooth conformal metric $g$ for just the same reason as in (\ref{eqn-conf-change-to-smooth}).
\end{def7}

\subsection{Comparison with Cut-off Methods}

   There are existing methods in the literature that essentially define the same quantity as (\ref{eqn-def-renorm-polya}), but instead of removing the balls~$B_{\varepsilon}(z_i,\tilde{g})$ entirely, they introduce a regularized metric inside. Such ideas are found in the incomplete notes of Zamolodchikov et al. \cite{ZZ} (page 86), as well as in Eskin, Kontsevich, and Zorich \cite{EKZ}, section 3.6. The cut-off method used in the latter is more mathematically rigorous, being smooth, and it avoids generating additional singularities along each~$\partial B_{\varepsilon}$ for~$\nabla \sigma$. In this subsection, we note that our definition \ref{def-renom-part-func} yields the same result as in \cite{EKZ}, and arguably, it also aligns with other regularization methods discussed in \cite{EKZ}, section 3.6.\\

For each fixed~$\delta> \delta'> 0$ we consider a smooth function~$\sreg_{\delta,\delta'}(r)$ such that
	\begin{equation}
	  \sreg_{\delta,\delta'}(r)=\left\{
	  \begin{array}{ll}
	   \log  r &\textrm{for }r\ge \delta,\\
	    \log \delta &\textrm{for }0\le r\le \delta',
	  \end{array}
	  \right.
	  \label{}
	\end{equation}
as well as
 	\begin{equation}
 \log \delta' \le  \sreg_{\delta,\delta'}(r) \le \log \delta, \quad \quad\textrm{for }0\le r\le \delta.
\end{equation}
and 
 	\begin{equation}
 \left| \partial_r \sreg_{\delta,\delta'}(r)  \right| \le C \delta'^{-1}, \quad\quad\textrm{for }0\le r\le \delta,
\end{equation}
for some constant $C$ independent of $\delta, \delta'$. 
Such a function can be constructed using a smooth cut-off function~$f:\mb{R}\lto\mb{R}$ for which~$f(t)\equiv 0$ for~$t\le 0$,~$0<f(t)<1$ for~$0<t<1$,~$f(t)\equiv 1$ for~$t\ge 1$, and posing
 	\begin{equation}
\sreg_{\delta,\delta'}(r) \defeq  \log \delta + f \Big(\frac{r- \delta'}{\delta - \delta'} \Big) \log \frac{r}{\delta} .
\end{equation}
Now given a closed Riemann surface~$\Sigma$ and~$\tilde{g}$ a generalized conformal metric representing~$D=\sum_{j=1}^p \gamma_j z_j$ with~$\gamma_j>-1$, as well as a smooth background conformal metric~$g$ such that around each singularity~$z_j$ holds (\ref{eqn-conic-sing-local-form}), and the regular metric potentials~$\varphi_j$ satisfy the Troyanov conditions (\ref{eqn-scale-metric-pot}), we introduce the regularized smooth metric~$g_{\varepsilon,\varepsilon'}$ (for~$\varepsilon> \varepsilon'> 0$ sufficiently small) 
   	\begin{equation}
	  g_{\varepsilon,\varepsilon'}\defeq \left\{
	  \begin{array}{ll}
	    \tilde{g}&\textrm{on }\Sigma\setminus\bigcup_1^p U_j,\\
	   \exp \big( 2 \gamma_j \cdot \sreg_{\delta_j,\delta'_j}(r_j) \big) \me^{2\varphi_j}g& \textrm{on }U_j,
	  \end{array}
	  \right.
	  \label{}
	\end{equation}
    where $r_j = d_g(\bullet,z_j)$, the~$U_j$'s given in definition \ref{def-singularities}, and we choose~$\delta_j = \delta_j(\varepsilon)$ to be the radius of the ball~$B_{\varepsilon}(z_j,\tilde{g})$ under the metric~$g$, and likewise $\delta'_j = \delta_j(\varepsilon')$.
	\begin{lemm}
	  In the setting as above, write~$g_{\varepsilon,\varepsilon'}=\me^{2\sigma_{\varepsilon,\varepsilon'}}g$ now with~$\sigma_{\varepsilon,\varepsilon'}\in C^{\infty}(\Sigma)$ for each~$\varepsilon>0$, and assuming $\varepsilon' \sim \varepsilon$ as $\varepsilon\to 0^+$, we have
	  \begin{equation}
	     \frac{1}{24\pi}\lim_{\varepsilon\to 0^+}\Big[ 
	     \int_{\Sigma}(|\nabla_g \sigma_{\varepsilon,\varepsilon'}|_g^2+2K_g\sigma_{\varepsilon,\varepsilon'})\,\dd V_g 
         +2\pi\sum_{i=1}^p \frac{\gamma_i^2}{1+\gamma_i}\log(\varepsilon)
         \Big] =\mathcal{R}A_{\Sigma}(\tilde{g},g).
	    \label{}
	  \end{equation}
	\end{lemm}

\begin{proof}
	  Since~$g_{\varepsilon,\varepsilon'}=\tilde{g}$ on~$\Sigma\setminus \bigcup_{i=1}^p B_{\varepsilon}(z_i,\tilde{g})$ we just need to show that for each~$j$,
	  \begin{equation}
	    \int_{B_{\varepsilon}(z_j,\tilde{g})}(|\nabla_g \sigma_{\varepsilon,\varepsilon'}|_g^2+2K_g\sigma_{\varepsilon,\varepsilon'})\,\dd V_g \quad\xlongrightarrow{\varepsilon\to 0^+}\quad 0.
	    \label{}
	  \end{equation}
  We suppress the subscript~$j$ as we treat each ball individually. Now
	  \begin{equation}
	    \sigma_{\varepsilon,\varepsilon'}(r)= \gamma \cdot\sreg_{\delta,\delta'}(r) +\varphi,
	    \label{}
	  \end{equation}
	  with~$r=d_g(\bullet,z_j)$. From the fact that 
 	\begin{equation}
 \log \delta' \le  \sreg_{\delta,\delta'}(r) \le \log \delta 
\end{equation}
one finds that 
 	\begin{equation}
\sigma_{\epsilon,\epsilon'}(r) = \mathcal{O}( \log \epsilon)
\end{equation}
  and hence the integral of~$K_g \sigma_{\varepsilon,\varepsilon'}$ over the ball of radius $\epsilon$ goes to zero.  Since~$\varphi$ satisfies the Troyanov conditions (\ref{eqn-scale-metric-pot}) we have
  \begin{equation}
    \int_{B_{\varepsilon}(z_j,\tilde{g})}|\nabla_g \varphi|_g^2\,\dd V_g \lesssim \int_0^{2\pi}\int_{0}^{\delta}r^{4\gamma+2}\cdot r\,\dd r\,\dd \theta=\mathcal{O}(\delta^{4\gamma+4})\quad\xlongrightarrow{\varepsilon\to 0^+}\quad 0,
    \label{}
  \end{equation}
  taking $(r,\theta)$ to be the geodesic polar coordinates with respect to $g$. Now it is straightforward to check that 
  	 \begin{equation}
    \sup_{r\le \delta}\big|\nabla_g \sreg_{\delta,\delta'}(r)\big|_{g} =\mathcal{O}(\delta^{-1}),\quad\quad \delta\to 0^+.
    \label{}
  \end{equation}
  provided $\delta' \sim \delta$.  Now together with (\ref{eqn-scale-metric-pot}) we have
  \begin{equation}
    \int_{B_{\varepsilon}(z_j,\tilde{g})}\ank{\nabla_g \sreg_{\delta,\delta'},\nabla_g\varphi}_g\,\dd V_g \lesssim \delta^{-1}\int_0^{2\pi}\int_{0}^{\delta}r^{2\gamma+1}\cdot r\,\dd r\,\dd \theta=\mathcal{O}(\delta^{2\gamma+2})\quad\xlongrightarrow{\varepsilon\to 0^+}\quad 0,
    \label{}
  \end{equation}
	  and finally, again by the fact that $\delta'\sim\delta$,
	  \begin{equation}
	    \int_{B_{\varepsilon}(z_j,\tilde{g})} \big|\nabla_g \sreg_{\delta,\delta'}(r)\big|^2\,\dd V_g \lesssim \frac{1}{\delta^2}\int_{\{\delta'\le r\le \delta\}}^{} \dd V_g
	    =\frac{1}{\delta}\mathcal{O}(\delta - \delta')\quad\xlongrightarrow{\delta \to 0^+}\quad 0. \label{eqn-comparison-cut-off-conv}
	  \end{equation}
	  which concludes the proof.
	\end{proof}

    \begin{def7}
    This lemma specifically explains why no boundary terms are included in the anomaly formula from definition \ref{def-renorm-polya-conic}. We are working over~$\Sigma\setminus \bigcup_{i=1}^p B_{\varepsilon}(z_i,\tilde{g})$, which is a surface with boundary, and the boundary is generally not geodesic. The reason we don’t include boundary terms is that we could have alternatively opted for the method described in this subsection, which only considers the closed surface~$\Sigma$, and still arrived at the same outcome. Moreover, as mentioned in the "future work" section of the introduction, we intend to investigate various settings involving surfaces with boundary in a future revision of the manuscript.
	\end{def7}

\begin{def7}
	  Imaginably, one could allow even more flexibility in the choice of the cut-off function~$\sreg_{\delta,\delta'}$ to gain even better convergence in (\ref{eqn-comparison-cut-off-conv}). Clearly one just needs to make the infinitesimal annuli with nonzero gradient~$\nabla \sigma_{\varepsilon,\varepsilon'}$ thin enough. Nevertheless our choice suffices.
	\end{def7}

\section{Main Proof}

\subsection{Proof}\label{sec-main-proof}

\begin{prop}\label{prop-main-conical-scaling}
	  Let~$\Sigma$ be a closed Riemann surface and~$\tilde{g}$ a generalized conformal metric representing~$D=\sum_{j=1}^p \gamma_j z_j$ with~$\gamma_j>-1$. Now suppose~$h\in C^{\infty}(\Sigma)$ and consider the scaled metric~$\me^{2h}\tilde{g}$ on~$\Sigma\setminus\supp D$. Then 
      \begin{equation}
	 \mathcal{R}A_{\Sigma}(e^{2h}\tilde{g},g) -\mathcal{R}A_{\Sigma}\left(\tilde{g},g\right)= \mathcal{R}A_{\Sigma}(e^{2h}\tilde{g},\tilde{g})  - \frac{1}{12} \sum_j   \frac{\gamma_j(\gamma_j+2)}{\gamma_j+1} h(z_j),
     \label{eqn-main-gen-conic-interp-cocycle}
	  \end{equation}
	  where $\mathcal{R}A_{\Sigma}(e^{2h}\tilde{g},\tilde{g}) $ is defined by (\ref{eqn-def-renorm-polya-conic}) and $\mathcal{R}A_{\Sigma}(e^{2h}\tilde{g},g) $, $\mathcal{R}A_{\Sigma}\left(\tilde{g},g\right)$ by definition \ref{def-main-renorm-anomaly}.
	\end{prop}

    \begin{def7}\label{rem-conic-cocycle}
  From this we deduce that if we consider a conformal field theory with central charge~$c$ defined on $\Sigma$ equipped with the metric $\tilde{g}$ as above, whose partition function $\mathcal{Z}(\Sigma,\tilde{g})$ is defined in definition \ref{def-renom-part-func}, then we have
  \begin{equation}
	    \frac{\mathcal{Z}(\Sigma, e^{2h} \tilde{g})}{\mathcal{Z}(\Sigma, \tilde{g})} = 
	    \exp\Big(c \, \mathcal{R}A_{\Sigma}(\me^{2h}\tilde{g},\tilde{g}) -\sum_{j=1}^p h(z_j)\Delta_{j} \Big), 
	    \label{eqn-main-general-conic}
	  \end{equation}
	  where
	  \begin{equation}
	    \Delta_{j} \defeq \frac{c}{12} \frac{\gamma_j(\gamma_j+2)}{\gamma_j+1}
	    \label{}
	  \end{equation}
	  are the \textsf{conformal weights} or \textsf{scaling dimensions} associated to conical singularities.
\end{def7}

    First we point out (directly from corollary \ref{cor-quad-blow-tilde})

\begin{lemm}\label{lemm-comp-renorm-anom}
  Let~$g$ be a smooth conformal metric on~$\Sigma$ so that near the cone points~$z_j$ the regular metric potentials~$\varphi_j$ of~$\tilde{g}$ against~$g$ satisfy (\ref{eqn-scale-metric-pot}).  Then
  \begin{equation}
    \mathcal{R}A_{\Sigma}(\tilde{g},g)=\frac{1}{12}\Big[-\sum_{j=1}^p\frac{ \gamma_j^2\log(\gamma_j+1)}{\gamma_j+1}+\sum_{j=1}^p \Big( \frac{\gamma_j^2}{\gamma_j+1}-\gamma_j \Big)\varphi_j(z_j)\Big]
    +\frac{1}{24\pi}\int_{\Sigma\setminus \supp D}^{} \sigma(-\Delta_g \sigma+2K_g)\,\dd V_g,
    \label{}
  \end{equation}
  where~$\tilde{g}=\me^{2\sigma}g$ with~$\sigma\in C^{\infty}(\Sigma\setminus \supp D)$.\hfill $\Box$
\end{lemm}

\begin{proof}
  [Proof of proposition \ref{prop-main-conical-scaling}.]  It follows from lemma \ref{lemm-consist} that it is sufficient to prove (\ref{eqn-main-gen-conic-interp-cocycle}) for a specific reference smooth metric $g$. Choose a smooth conformal metric~$g$ with the same conditions as in lemma \ref{lemm-comp-renorm-anom}. Write~$\tilde{g}=\me^{2\sigma}g$ then~$\me^{2h}\tilde{g}=\me^{2(h+\sigma)}g$,~$\sigma\in C^{\infty}(\Sigma\setminus \supp D)$. Using lemma \ref{lemm-comp-renorm-anom} and relation (\ref{eqn-conf-change-to-smooth}) one obtains immediately 
  \begin{equation}
	 \mathcal{R}A_{\Sigma}(e^{2h}\tilde{g},g) - \mathcal{R}A_{\Sigma}(e^{2h}\tilde{g},\tilde{g}) - \mathcal{R}A_{\Sigma}\left(\tilde{g},g\right)  = - \frac{1}{12}\sum_j   \frac{\gamma_j}{\gamma_j+1} h(z_j) + \frac{1}{24 \pi} \int_{\Sigma\setminus \supp D} \left( h \Delta_g \sigma -  \sigma \Delta_g h \right) \,\dd V_g.
	  \end{equation}
What we desire follows then from applying lemma \ref{lemm-app-poincare-lel} to $\sigma$ which has the local form (\ref{eqn-local-form-conf-factor}).
\end{proof}

\begin{def7}\label{rem-extra-quad-term}
  In view of the Poincar\'e-Lelong lemma it would be instructive to compare the result with a ``naive guess'' for the Polyakov formula by pretending that the conical singularities did no more than just putting a log term into the metric potential and hence creating a delta function in the curvature. This is to say we suppose
  \begin{equation}
    \log\Big(
    \frac{\mathcal{Z}(\Sigma, e^{2h} \tilde{g})}{\mathcal{Z}(\Sigma, \tilde{g})} \Big)\textrm{``}=\textrm{''}\frac{c}{24\pi}\int_{\Sigma}(|\nabla_{\tilde{g}} h|_{\tilde{g}}^2+2K_{\tilde{g}}h)\,\dd V_{\tilde{g}}.
    \label{}
  \end{equation}
  Now taking into account the delta functions produced by~$-\Delta_g \sigma$ in (\ref{eqn-conf-change-to-smooth}) over~$\supp D$, we find
  \begin{equation}
    \int_{\Sigma}(|\nabla_{\tilde{g}} h|_{\tilde{g}}^2+2K_{\tilde{g}}h)\,\dd V_{\tilde{g}} =\int_{\Sigma\setminus\supp D}(|\nabla_{\tilde{g}} h|_{\tilde{g}}^2+2K_{\tilde{g}}h)\,\dd V_{\tilde{g}} -4\pi\sum_{j=1}^p \gamma_j h(z_j).
    \label{}
  \end{equation}
  Therefore the extra piece that we get in the actual formula is the quadratic term~$\sum_j \frac{2\pi\gamma_j^2}{\gamma_j+1}h(z_j)$. Going back to the definitions, we could see that this term originates as the asymptotic Dirichlet energy of the metric potential~$\sigma$ on a shrinking ``scaled annulus'', which we singled out as lemma \ref{lemm-quadratic-sing} above (cf.\ corollary \ref{cor-ratio-radius}).
\end{def7}

\begin{prop}\label{prop-main-app-ram-cov}
	  Let~$\Sigma_d$,~$\Sigma$ be closed Riemann surfaces, with a smooth conformal metric~$g$ on~$\Sigma$, and~$f:\Sigma_d\lto \Sigma$ a ramified $d$-sheeted holomorphic map, whose critical values are~$w_1$, \dots,~$w_p$. Consider a conformal field theory with central charge~$c$ whose partition function is denoted~$\mathcal{Z}$. Pick~$h\in C^{\infty}(\Sigma)$ on~$\Sigma$. Then under the definition \ref{def-renom-part-func} for $\mathcal{Z}(\Sigma_d, f^* e^{2h} g)$ and $\mathcal{Z}(\Sigma_d, f^* g)$ we have
	  \begin{equation}
	    \frac{\mathcal{Z}(\Sigma_d, f^* e^{2h} g)}{\mathcal{Z}(\Sigma, e^{2h} g)^d} = e^{- \sum_j  h(w_j)\Delta_{j}} \frac{\mathcal{Z}(\Sigma_d, f^* g)}{\mathcal{Z}(\Sigma, g)^d}, 
	    \label{}
	  \end{equation}
	  where
	  \begin{equation}
	    \Delta_{j} \defeq \frac{c}{12}  \sum_{z \in f^{-1}(w_j)}\Big(\ord_f(z) - \frac{1}{\ord_f(z)} \Big)
	    \label{}
	  \end{equation}
	  is the \textsf{conformal weight} or \textsf{scaling dimension} associated to the point $w_j$.
	\end{prop}
\begin{proof}
	  Following proposition \ref{prop-main-conical-scaling}, the only thing we need to check is
	  \begin{equation}
      \mathcal{R}A_{\Sigma_d}\big(f^*\me^{2h}g,f^*g\big)= d \cdot A_{\Sigma}\big(\me^{2h}g,g \big),
	    \label{}
	  \end{equation}
      where we recall that the left hand side is defined by (\ref{eqn-def-renorm-polya-conic}). Indeed, let $\{z_1,\dots,z_q\}$ denote the preimage of the critical values $\{w_1,\dots,w_p\}$ of $f$. Then on~$\Sigma_d\setminus \{z_1,...,z_q\}$ the map~$f$ is an unbranched~$d$-fold covering, we have
\begin{equation}
  \int_{\Sigma_d\setminus \{z_1,...,z_q\}}^{}
  \big(|\nabla_{f^*g}f^*h|_{f^*g}^2 +2K_{f^*g}\cdot f^*h \big)\,\dd V_{f^*g}=d\cdot \int_{\Sigma\setminus \{w_1,...,w_p\}}^{}
\big(|\nabla_{g}h|_{g}^2 +2K_{g}\cdot h \big)\,\dd V_{g},
  \label{}
\end{equation}
 and we finish the proof.
	\end{proof}

    \begin{def7}\label{rem-diffeo-inv}
  Note that~$\mathcal{Z}(\Sigma_d, f^* g)/\mathcal{Z}(\Sigma, g)^d$ is diffeomorphism invariant as a function of the critical values of~$f$ because~$\mathcal{Z}(\Sigma_d, f^*(\Psi^* g))/\mathcal{Z}(\Sigma,\Psi^* g)^d=\mathcal{Z}(\Sigma_d, (\Psi\circ f)^* g)/\mathcal{Z}(\Sigma, g)^d$ if~$\Psi:\Sigma\lto \Sigma$ is a diffeomorphism, for the ordinary~$\mathcal{Z}(\Sigma,g)$ is so invariant by definition.
\end{def7}

    \begin{corr}
  Consider a quantum system in the sense of subsection \ref{sec-main-replica} defined on a circle~$\mb{S}^1_L$ of perimeter~$L$, with central charge~$c$. Let~$A\subset \mb{S}^1_L$ be an interval of length~$\ell$ and let~$\rho$ be the \textsf{vacuum state} associated to this system. Then under the \textsf{replica interpretation} we have
  \begin{equation}
    \ttr_{\mathcal{H}_A}\big( \ttr_{A^c}(\rho)^d\big) =C \Big( \frac{L}{\pi}\sin \frac{\pi\ell}{L} \Big)^{-\frac{c}{6}\left( d-\frac{1}{d} \right)},
    \label{}
  \end{equation}
  where~$C$ is the same constant as in lemma \ref{lemm-two-pt-func}.
\end{corr}

\begin{proof}
  We adopt the interpretation of example \ref{exp-ground-state-disj-cob}. Then the surface~$\check{\Sigma}$ corresponds to a re-scaled Fubini-Study Riemann sphere (of radius~$L/2\pi$) where~$\mb{S}^1_L$ embeds isometrically as the equator. By proposition \ref{prop-main-app-ram-cov} we could first assume~$L=2\pi$. Suppose the end points of~$A$ correspond to~$0$ and~$z_0\in \wh{\mb{C}}$. Then the surface~$\check{\Sigma}_d$ is topologically still~$\wh{\mb{C}}$ but equipped with the pull-back of the Fubini-Study metric under a holomorphic map~$f_d:\wh{\mb{C}}\lto \wh{\mb{C}}$ which ramifies exactly at~$0$ and~$z_0$ both with order~$d$. Such a map could be given by (but not related to final result)
  \begin{equation}
    f_d(z)\defeq \frac{z_0\big( \frac{z}{z-z_0} \big)^d}{\big( \frac{z}{z-z_0} \big)^d -1}=\frac{z_0 z^d}{z^d-(z-z_0)^d}.
    \label{}
  \end{equation}
  Now we apply a smooth scaling~$\me^{2h}$ with~$h=\log(L/ 2\pi)$ on the equator, while making $\me^{2h}g_{\mm{FS}}$ flat in a thin neighborhood of the equator, in accordance with the requirement for vacuum state explained in example \ref{exp-ground-state-disj-cob}. The latter is always possible for Fubini-Study since~$4(1+|z|^2)^{-2}=|z|^{-2}+\mathcal{O}((1-|z|)^2)$. Then by proposition \ref{prop-main-app-ram-cov}, lemma \ref{lemm-two-pt-func} (and remark \ref{rem-diffeo-inv}) we obtain
  \begin{align*}
    \ttr_{\mathcal{H}_A}\big( \ttr_{A^c}(\rho)^d\big)&=\frac{\mathcal{Z}(\check{\Sigma}_d)}{\mathcal{Z}(\check{\Sigma})^d}=C\me^{-h(0)\Delta_d}\me^{-h(z_0)\Delta_d}\sin\Big( \frac{1}{2}d_{\mm{FS}}(0,z_0) \Big)^{-2\Delta_d}\\
    &=C\Big( \frac{L}{\pi}\sin \frac{\pi\ell}{L} \Big)^{-2\Delta_d},
  \end{align*}
  with~$\Delta_d=\frac{c}{12}(d-\frac{1}{d})$, giving us the result.
\end{proof}

\subsection{Comments on Relation to Literature}\label{sec-literature}
\begin{def7}
  [comparison with a result of V.\ Kalvin \cite{Kalvin}] \label{rem-rel-kal} We note here that V.\ Kalvin has obtained a closely related result in \cite{Kalvin} corollary 1.3(1) for the~$\zeta$-determinant of the Friederichs Laplacian under the conical metric with zero ``boundary condition'' at the cone points. As far as the Polyakov formula is concerned, the~$\zeta$-determinant is no different from a general CFT partition function in the sense of definition \ref{def-cft}, except for a constant. So our proposition \ref{prop-main-conical-scaling} coincides with the case of his result when the two divisors coincide. If they don't coincide, one could readily derive the corresponding result out of our method following lemma \ref{lemm-comp-renorm-anom}. However, in that case, the constants (last term in \cite{Kalvin} eqn.\ (1.8)) would be different, as our lemma \ref{lemm-comp-renorm-anom} only agrees with \cite{Kalvin} eqn.\ (1.1) upto the integral term and the term involving the regular metric potential. Supposedly, the extra constants in \cite{Kalvin} come from the specific Friederichs extension.

  On a deeper level, the proof of \cite{Kalvin} proceeds also by removing disks but then using the BFK decomposition formula. One key point is to show that the blow-up order of~$\log\detz \Delta_{D_{\varepsilon}(z_j)}$ of the Friederichs Dirichlet Laplacian~$\Delta_{D_{\varepsilon}(z_j)}$ on the shrinking disks~$D_{\varepsilon}(z_j)$ centered at the cone points with the singular metric, added up, matches exactly with the blow-up order of the~$\log\detz$ of the Dirichlet Laplacian on their complement, as the disks shrink (as the contributions from D-to-N maps cancel). In this regard we keenly observe that the cancellation mechanism as shown in \cite{Kalvin} eqn.\ (2.28) is \textit{different} from ours. The radii of Kalvin's disks were measured under the background smooth metric. For us, on the other hand, it is important to align the rate of shrinking with the singular metric itself (i.e.\ measure the disks using the singular metric) in two respects. First to obtain the invariance of the definition of the renormalized partition function as in lemma \ref{lemm-consist}, and secondly to obtain the correct scaling dimensions in the final result. A more detailed understanding of the relation of the two methods seems desirable, especially considering that the BFK formula is closely related to Segal's gluing axioms for QFT/CFT \cite{Lin, GKRV}, and that the correlation functions are (conjecturally) the ``limits'' of propagators associated to the complements of the shrinking disks (\cite{KS} section 3).
\end{def7}

\begin{def7}
  [relation to Aldana, Kirsten and Rowlett \cite{AKR}] \label{rem-rel-akr} As the authors have themselves remarked in relation to \cite{Kalvin}, they also obtained the same result as above for the~$\zeta$-determinant (\cite{AKR} eqn.\ (1.5)). Their method is based on asymptotic analysis of the heat trace for which related results were obtained earlier \cite{MR}. For the latter, they employed a blow-up (geometric) technique and the cone contribution term (main interest of this article) seems to ultimately come from computations done to an infinite exact cone which has a longer history (see their appendix A for further references). Last but not least, we note that \cite{AKR} seems to have restricted their cone angles to~$(0,2\pi)$, which is in the complement of the main concern of this article (proposition \ref{prop-main-app-ram-cov}).
\end{def7}

\appendix

\newpage

\section{Appendix: Poincar\'e-Lelong Lemmas}\label{sec-app-poin-lelong}

\begin{proof}[Proof of lemma \ref{lemm-bounded-lap-log}.]
	  We denote by~$\nabla_v$ the Levi-Civita covariant derivative in the direction~$v$. Since we are concerned with only one metric~$g$ we omit it from the notations. Take the geodesic polar coordinates $(r, \theta)$ around $z_0$. Then
 $g = \dd r^2 + \omega^2(r) \dd\theta^2$ where $\omega(r) = | \partial_{\theta}|$. From the expression of $\Delta$ in polar coordinates we have
	  \begin{align*}
	    \Delta \log r =  \frac{1}{\omega(r)} \partial_r \left (\omega(r) \partial_r \log r \right) = -\frac{1}{r^2} + \frac{\partial_r \omega(r)}{r\, \omega(r)}.
	  \end{align*}
	  Either noting the fact that~$\partial_{\theta}$ is a Jacobi field along $\theta=\textrm{const}$ connecting~$z_0$ to~$z$ (\cite{dC} page 115), or otherwise, we have a Taylor expansion
	  \begin{equation}
	    |\partial_{\theta}|=r+\mathcal{O}(r^3),\quad r\to 0.
	    \label{}
	  \end{equation}
	  This gives
	  \begin{equation}
	    \frac{\partial_r|\partial_{\theta}|}{|\partial_{\theta}|}=\frac{1}{r}+\mathcal{O}(r),\quad r\to 0,
	    \label{}
	  \end{equation}
	  yielding the result.
	\end{proof}

\begin{lemm}\label{lemm-app-poincare-lel}
  Let~$(\Sigma,g)$ be a Riemannian surface with smooth metric~$g$, and fix~$z_0\in \Sigma$. Denote~$r(z):=d_g(z,z_0)$. Then for any~$\psi\in C_c^{\infty}(U)$ where the neighborhood~$U$ lies within the injectivity radius of~$z_0$,
  \begin{equation}
    \int_{\Sigma}^{}\log r (-\Delta_g \psi)\,\dd V_g=\int_{\Sigma\setminus z_0}^{}\psi (-\Delta_g \log r)\,\dd V_g -2\pi \psi(z_0).
    \label{}
  \end{equation}
\end{lemm}

\begin{proof}
  Pick the metric disk~$B_{\varepsilon}(z_0)$ with radius~$\varepsilon>0$ at~$z_0$. Then Green's formula gives
  \begin{align*}
    \int_{\Sigma\setminus B_{\varepsilon}(z_0)}^{}\big[\log r (\Delta_g \psi)-\psi (\Delta_g \log r)\big]\,\dd V_g&= -\int_{\partial B_{\varepsilon}}^{} \big[\log r (\partial_r \psi)-\psi (\partial_r \log r)\big]\,\dd \ell_g\\
    &=\mathcal{O}(\varepsilon\log \varepsilon)+\psi(z_0)\frac{1}{\varepsilon}2\pi \varepsilon +\mathcal{O}(\varepsilon),
  \end{align*}
  as~$\varepsilon\to 0$, since~$\psi$ is smooth and (\ref{eqn-asymp-peri-circle}).
\end{proof}

\begin{def7}
  More precisely we have the equality of currents~$\Delta_g \log r\,\dd V_g=\Delta_g \log r|_{U\setminus z_0}\dd V_g+2\pi \delta_{z_0}$ on~$U$.
\end{def7}

\begin{lemm}\label{lemm-bound-lap-log-ratio}
  Fix~$r(z):=d_g(z,z_0)$ and now let~$g_1=\me^{2h}g$,~$h\in C^{\infty}(\Sigma)$, and~$r_1(z):=d_{g_1}(z,z_0)$. Then
  \begin{equation}
    \Delta_{g_1}\log r -\Delta_{g_1}\log r_1 \in L^1(U,g_1).
    \label{}
  \end{equation}
  Namely, as a distribution, it agrees with an integrable function when paired with~$\dd V_{g_1}$.
\end{lemm}

\begin{proof}
  The point is that the current~$\Delta_{g_1}\log r\,\dd V_{g_1}$ agrees with~$\Delta_g \log r\,\dd V_g$, because of their definitions and conformal covariances (\ref{eqn-scale-vol}), (\ref{eqn-scale-lap}). Therefore the delta functions cancel out exactly and we are left with the regular parts~$\Delta_{g_1}\log r|_{U\setminus z_0} -\Delta_{g_1}\log r_1|_{U\setminus z_0}$ which is in fact~$L^{\infty}(U,g_1)$.
\end{proof}


\newpage


\begin{thebibliography}{99}

\addcontentsline{toc}{section}{References}

\bibitem{AKR}
Clara L. Aldana, Klaus Kirsten, Julie Rowlett, \textit{Polyakov formulas for conical singularities in two dimensions}, arXiv preprint \href{https://arxiv.org/abs/2010.02776}{arXiv:2010.02776}.

\bibitem{Knizhnik}
V. G. Knizhnik. \textit{Analytic fields on Riemann surfaces. II.} Communications in Mathematical Physics, 112(4):567–590, 1987.

\bibitem{Dixon}
Lance Dixon, Daniel Friedan, Emil Martinec and Stephen Shenker. \textit{The conformal field theory of orbifolds}. Nuclear Physics B Volume 282, 1987, Pages 13-73

\bibitem{Holzhey}
Christoph Holzhey, Finn Larsen an Frank Wilczek. \textit{Geometric and renormalized entropy in conformal field theory}. Nuclear Physics B Volume 424, Issue 3, 15 August 1994, Pages 443-467

\bibitem{HRV}
Hubert Lacoin; Rémi Rhodes; Vincent Vargas. \textit{The semiclassical limit of Liouville conformal field theory}. Annales de la Faculté des sciences de Toulouse : Mathématiques, Série 6, Tome 31 (2022) no. 4, pp. 1031-1083. \url{https://afst.centre-mersenne.org/articles/10.5802/afst.1713/}

\bibitem{Cardy_Peschel}
Cardy, John L. and Peschel, Ingo, \textit{Finite Size Dependence of the Free Energy in Two-dimensional Critical Systems}, Nuclear Physics B Volume 300, 1988, Pages 377-392

\bibitem{Cardy_Calabrese} 
Pasquale Calabrese and John Cardy. \textit{Entanglement entropy and conformal field theory}.  Journal of Physics A: Mathematical and Theoretical, Volume 42, Number 50

\bibitem{Wiegmann}
T. Can and P. Wiegmann, \textit{Quantum Hall states and conformal field theory on a singular surface},  Journal of Physics A: Mathematical and Theoretical, Volume 50, Number 49


\bibitem{Ahlfors}
Lars Ahlfors, \textit{Complex Analysis: An Introduction to the Theory of Analytic Functions of One Complex Variable}, Third Edition, Volume: 385, AMS Chelsea Publishing, 2021.

\bibitem{Cheeger}
Jeff Cheeger, \textit{Spectral geometry of singular Riemannian spaces}, Journal of Differential Geometry, J. Differential Geom. 18(4), 575-657, (1983)

\bibitem{DW}
Xianzhe Dai and Guofang Wei, \textit{Comparison Geometry for Ricci Curvature}, lecture notes, available at \url{https://web.math.ucsb.edu/~dai/Ricci-book.pdf}.

\bibitem{dC}
Manfredo P. do Carmo, \textit{Riemannian Geometry}, Mathematics: Theory \& Applications, Birkhäuser, Boston, 1992.

\bibitem{dC2}
Manfredo P. do Carmo, \textit{Differential Geometry of Curves and Surfaces}, Revised and Updated Second Edition, Courier Dover Publications, 2016.

\bibitem{EKZ}
Eskin, A., Kontsevich, M. and Zorich, A. \textit{Sum of Lyapunov exponents of the Hodge bundle with respect to the Teichmüller geodesic flow}. Publ.math.IHES 120, 207–333 (2014).

\bibitem{Gaw}
Krzysztof Gaw\c{e}dzki. Lectures on conformal ﬁeld theory. In \textit{Quantum ﬁelds and 
strings: A course for mathematicians}, pages 727–805. American Mathematical 
Society, 1999.

\bibitem{GKRV}
Guillarmou, C., Kupiainen, A., Rhodes, R., and Vargas, V. (2021). 
\textit{Segal's axioms and bootstrap for Liouville Theory.} arXiv preprint arXiv:2112.14859.

\bibitem{Hea}
Matthew Headrick, \textit{Lectures on entanglement entropy in field theory and holography}, arXiv preprint \href{https://arxiv.org/abs/1907.08126}{1907.08126}.

\bibitem{HT}
Hulin, D., Troyanov, M. \textit{Prescribing curvature on open surfaces}. Math. Ann. 293, 277–315 (1992). \url{https://doi.org/10.1007/BF01444716}

\bibitem{JV}
Kurt Johansson, Fredrik Viklund, \textit{Coulomb gas and the Grunsky operator on a Jordan domain with corners}, \href{https://arxiv.org/abs/2309.00308}{arXiv:2309.00308}.

\bibitem{Kalvin}
Victor Kalvin,
\textit{Polyakov-Alvarez type comparison formulas for determinants of Laplacians on Riemann surfaces with conical singularities},
Journal of Functional Analysis,
Volume 280, Issue 7,
2021. \url{https://doi.org/10.1016/j.jfa.2020.108866}


\bibitem{KMMW}
Klevtsov, S., Ma, X., Marinescu, G. et al. \textit{Quantum Hall Effect and Quillen Metric}. Commun. Math. Phys. 349, 819–855 (2017). \url{https://doi.org/10.1007/s00220-016-2789-2}



\bibitem{KS}
Maxim Kontsevich, Graeme Segal, \textit{Wick Rotation and the Positivity of Energy in Quantum Field Theory}, The Quarterly Journal of Mathematics, Volume 72, Issue 1-2, June 2021, Pages 673–699, \url{https://doi.org/10.1093/qmath/haab027}



\bibitem{Lin}
Jiasheng Lin. \textit{The Bayes Principle and Segal Axioms for $P(\phi)_2$, with application to Periodic Covers}. arXiv preprint \href{https://arxiv.org/abs/2403.12804}{2403.12804}.

\bibitem{MR}
MAZZEO, RAFE, and JULIE ROWLETT. \textit{A Heat Trace Anomaly on Polygons}. Mathematical Proceedings of the Cambridge Philosophical Society 159, no. 2 (2015): 303–19. \url{https://doi.org/10.1017/S0305004115000365}

\bibitem{MRS}
Rafe Mazzeo, Yanir Rubinstein, Natasa Sesum, \textit{Ricci flow on surfaces with conic singularities}, Analysis \& PDE, Anal. PDE 8(4), 839-882, (2015)

\bibitem{PSa}
Ralph Phillips, Peter Sarnak, \textit{Geodesics in homology classes}, Duke Math. J. 55(2): 287-297 (June 1987).

\bibitem{PW}
Eveliina Peltola, Yilin Wang, \textit{Large deviations of multichordal SLE 0+, real rational functions, and zeta-regularized determinants of Laplacians}. J. Eur. Math. Soc. 26 (2024), no. 2, pp. 469–535.

\bibitem{PWer} Powell, Ellen and Wendelin Werner, \textit{Lecture Notes on the Gaussian Free Field}, Cours Sp\'ecialis\'es, Vol. 28, Soci\'et\'e Math\'ematique de France, 2021.

\bibitem{RS} Ray, D. B., and I. M. Singer, \textit{R-Torsion and the Laplacian on Riemannian Manifolds}, Advances in Mathematics 7, 145-210, 1971.


\bibitem{RSim}
M. Reed, B. Simon,
\textit{Methods of Modern Mathematical Physics}, Volume I: Functional Analysis,
Academic Press,
1972.

\bibitem{RSim2}
M. Reed, B. Simon,
\textit{Methods of Modern Mathematical Physics}, Volume II: Fourier Analysis, Self-adjointness,
Academic Press,
1975.

\bibitem{RT}
Mukund Rangamani, Tadashi Takayanagi, \textit{Holographic Entanglement Entropy}, Lecture Notes in Physics vol.\ 931, Springer Cham, 2017. \url{https://doi.org/10.1007/978-3-319-52573-0}

\bibitem{Rudin} Rudin, Walter, \textit{Functional Analysis}, Second Edition, McGraw-Hill, Inc., 1991.

\bibitem{Segal} G. Segal, \textit{The definition of conformal field theory}, Topology, Geometry, and Quantum Field
Theory, (Ed. U. Tillmann), Vol. 308, London Math. Soc. Lecture Notes, London, 2004, 421–577.

\bibitem{Sim1} Simon, Barry, \textit{Trace Ideals and their Applications}, Second Edition, Mathematical Surveys and Monographs, Vol. 120, American Mathematical Society, 2005.

\bibitem{Sim2} Simon, Barry, \textit{The $P(\phi)_2$ (Euclidean) Quantum Field Theory}, Princeton University Press, New Jersey, 1974.

\bibitem{Sim3} Simon, Barry, \textit{Operator Theory: A Comprehensive Course in Analysis, Part 4}, Vol. 4, American Mathematical Society, 2015.

\bibitem{Sim4}
Barry Simon,
\textit{Notes on infinite determinants of Hilbert space operators},
Advances in Mathematics,
Volume 24, Issue 3,
1977,
Pages 244-273.


\bibitem{Troy}
Troyanov, Marc. \textit{Prescribing Curvature on Compact Surfaces with Conical Singularities}. Transactions of the American Mathematical Society 324, no. 2 (1991): 793–821.

\bibitem{Troy2}
Troyanov, Marc. \textit{Coordonnées polaires sur les surfaces riemanniennes singulières}. Annales de l'Institut Fourier, Tome 40 (1990) no. 4, pp. 913-937. \url{http://www.numdam.org/articles/10.5802/aif.1241/}

\bibitem{Witten}
Edward Witten. \textit{Introduction to Black Hole Thermodynamics}. arXiv preprint \href{https://arxiv.org/abs/2412.16795}{2412.16795}.

\bibitem{ZZ}
Alexei Zamolodchikov and Alexander Zamolodchikov, \textit{Lectures on Liouville Theory and Matrix Models}, lecture notes, available at: \url{https://qft.itp.ac.ru/ZZ.pdf}





\end{thebibliography}
\end{document}